\documentclass[12pt,reqno]{amsart}
 
\usepackage{geometry}
\usepackage[utf8]{inputenc}
\usepackage[english]{babel}
\usepackage{amssymb}
\usepackage{amsmath}
\usepackage{amsthm}
\usepackage{amscd}
\usepackage{enumitem}
\usepackage{graphicx}
\usepackage{tikz}
\usepackage{dsfont}
\usepackage{multicol}

\newcommand{\eps}{\varepsilon}

\def\d{\mathrm{d}}

\def\R{\mathbb R}

\def\esssup{\mathrm{ess\, sup\,}}
\def\essinf{\mathrm{ess\, inf\,}}
\def\e{\mathrm{e}}
\def\d{\mathrm{d}}
\def\VaR{\mathrm{VaR}}
\def\TVaR{\mathrm{TVaR}}

\theoremstyle{plain}
\newtheorem{theorem}{Theorem}[section]
\newtheorem{lemma}[theorem]{Lemma}
\newtheorem{corollary}[theorem]{Corollary}
\newtheorem{proposition}[theorem]{Proposition}
\newtheorem{problem}[theorem]{Problem}

\theoremstyle{definition}
\newtheorem{definition}[theorem]{Definition}
\newtheorem{example}[theorem]{Example}

\theoremstyle{remark}
\newtheorem{remark}[theorem]{Remark}

\numberwithin{equation}{section}
\numberwithin{figure}{section}

\title{Orlicz-Lorentz premia and distortion Haezendonck-Goovaerts risk measures}
\author{Aline Goulard and Karl Grosse-Erdmann}
\address{Aline Goulard, 
D\'epartement de Math\'e\-matique, Universit\'e de Mons, 20 Place du Parc, 7000 Mons, Belgium}
\email{alinegoulard@gmail.com}
\address{Karl Grosse-Erdmann, 
D\'epartement de Math\'e\-ma\-tique, Universit\'e de Mons, 20 Place du Parc, 7000 Mons, Belgium}
\email{kg.grosse-erdmann@umons.ac.be}


\begin{document}

\begin{abstract}
In financial and actuarial research, distortion and Haezendonck-Goovaerts risk measures are attractive due to their strong properties. They have so far been treated separately. In this paper, following a suggestion by Goovaerts, Linders, Van Weert, and Tank, we introduce and study a new class of risk measure that encompasses the distortion and Haezendonck-Goovaerts risk measures, aptly called the distortion Haezendonck-Goovaerts risk measures. They will be defined on a larger space than the space of bounded risks. We provide situations where these new risk measures are coherent, and explore their risk theoretic properties.
\end{abstract}
\maketitle

\section{Introduction}\label{s-intro}
Risk measures occupy a prominent role in financial and actuarial research, see \cite{DDGK05}, \cite{FoSc16}, \cite{RoSg13}, and \cite{Rue13}. The most basic risk measure is Value at Risk VaR$_\alpha$, $0<\alpha\leq 1$, which is simply the quantile of order $\alpha$ of a given risk $X$: VaR$_\alpha(X)=F_X^{-1}(\alpha)$. Once it was recognized that VaR does not satisfy the desirable property of subadditivity (but see the discussion in \cite{DLVDG08}), more advanced risk measures were proposed and studied. The best known subadditive alternative to VaR is the Tail Value at Risk TVaR$_\alpha$, $0<\alpha< 1$, also known as Expected Shortfall, Average Value at Risk or Conditional Value at Risk, which is a weighted (or distorted) version of VaR. Using different weight functions, one is led to the large and well-studied family of distortion risk measures, defined by
\[
\rho_g(X)= \int_0^1 F_X^{-1}(1-u)\d g(u),
\]
where $g$ is a distortion function. The literature on these risk measures is extensive, see for example \cite{AmLi24}, \cite{DKLT12}, \cite{DVGKTV06}, \cite{GLVT12}, and \cite{Wa96}; see also \cite{HJZ15} and \cite{Wei18}, where they are called weighted VaR.

A different class of risk measures is based on the idea of applying a convex function $\phi$ (more precisely, a Young function) to VaR. Inspired by the theory of Orlicz spaces, Haezendonck and Goovaerts \cite{HaGo82} defined a corresponding Orlicz premium for positive risks $X$, see Definition \ref{d-Orlicz}; it may be defined equivalently as
\[
\pi_{\phi,\alpha}(X) = \inf\Big\{a>0 : \int_0^1\phi\Big(\frac{F_X^{-1}(1-u)}{a}\Big)\d u\leq 1-\alpha\Big\},
\]
where $\alpha<1$. The extension to real-valued risks in a cash-invariant way was subsequently proposed by Goovaerts, Kaas, Dhaene, and Tang \cite{GKDT04} as $\rho_{\phi,\alpha}(X) = \inf_{x\in\mathbb{R}}(\pi_{\phi,\alpha}((X-x)^+)+x)$. These so-called Haezendonck-Goovaerts risk measures (see \cite[p.\ 13]{GLVT12}) have been studied intensively, see for example \cite{AhSh14}, \cite{AmLi24}, \cite{BeRo08}, \cite{BeRo08b}, \cite{BeRo12}, \cite{CCR21}, \cite{GMX20}, \cite{GKDT04}, \cite{GLVT12}, \cite{LPW17}, and \cite{TaYa14}.

It therefore seems natural and of interest to combine these two ways of weighting VaR. This was suggested, en passant, by Goovaerts, Linders, Van Weert, and Tank \cite[Definition 4.2]{GLVT12}. Analysing their suggestion leads us to the premium
\begin{align*}
\pi_{g,\phi,\alpha}(X)= \inf\Big\{a>0 : \int_0^1\phi\Big(\frac{F_X^{-1}(1-u)}{a}\Big)\d g(u)\leq 1-\alpha\Big\},
\end{align*}
which we call an Orlicz-Lorentz premium in view of its link with the Orlicz-Lorentz spaces, and to the distortion Haezendonck-Goovaerts risk measure 
\[
\rho_{g,\phi,\alpha}(X) = \inf_{x\in\mathbb{R}}(\pi_{g,\phi,\alpha}((X-x)^+)+x).
\]
The main aim of our paper is to determine natural sets where these risk measures are defined, and to study their risk theoretic properties. Our main result is that the distortion Haezendonck-Goovaerts risk measures are coherent whenever $g$ is concave, thereby generalizing the known properties for distortion and Haezendonck-Goovaerts risk measures.

The large majority of our results were first presented in 2022 in the PhD thesis of the first author \cite{Gou22}. The main additional contributions are the investigation of Fatou properties, the realization that the Orlicz-Lorentz premia are closely related to the Orlicz-Lorentz spaces from functional analysis (hence their name), and the observation that, in many cases, Haezendonck-Goovaerts risk measures reduce to the expectation when $\alpha=0$. Also, we offer a different proof of coherence: while in \cite{Gou22}, the proof was more direct, we proceed here via the notions of stop-loss order and comonotonicity, as suggested in \cite{DVGKTV06} and \cite{WaDh98}.

The paper is organized as follows. In Section \ref{s-rm} we recall the main risk theoretic properties that are discussed in this paper.
Sections \ref{s-dist} and \ref{s-HaezGoov} present the distortion risk measures and the Haezendonck-Goovaerts risk measures, respectively; they prepare the ground for the following section, but they also add some new aspects to the known theory, like Example \ref{ex-cone}, Proposition \ref{pex-HaezGoov}, and the unexpected Corollary \ref{c-HGexp}. Section \ref{s-distHaezGoov} constitutes the main part of this paper, a thorough investigation of the distortion Haezendonck-Goovaerts risk measures. 

We remark that recently, and independently, Wu and Xu \cite{WuXu22} have also proposed versions of the Orlicz-Lorentz premia and the distortion Haezendonck-Goovaerts risk measures. We discuss the relationship with our work in the final Section \ref{s-end}. We also suggest there some open problems.

Let us finally mention that properties like ``positive'' and ``decreasing'' are meant in the large sense. Also, random variables that coincide almost surely are identified. Thus, for example, ``$X\geq Y$'' means that ``$X\geq Y$ a.s.'' As usual, $L^\infty$ and $L^1$ denote the spaces of bounded and of integrable random variables, respectively, while $L^\infty_+$ and $L^1_+$ are their positive cones. We emphasize that $\esssup X$ is defined for any random variable, having the value $\infty$ if $X$ is not bounded above. The following well-known properties of the quantile function $F_X^{-1}(u)= \inf\{x\in\mathbb{R}: F_X(x)\geq u\}$ will be used repeatedly. If $h$ is a continuous increasing function on $\mathbb{R}$ then $F^{-1}_{h(X)}=h(F^{-1}_X)$; if $h$ is a positive measurable function on $\mathbb{R}$ then $\int_\Omega h(X)\d P=\int_0^1 h(F_X^{-1}(u))\d u$; and $u\leq F_X(x)$ holds if and only if $F_X^{-1}(u)\leq x$.

\section{Risk measures}\label{s-rm}
Throughout this paper, risk variables $X$ are real random variables on a given probability space $(\Omega,\mathcal{A},P)$. We follow the usual convention from insurance mathematics: positive values of $X$ correspond to losses, negative ones correspond to gains.

\begin{definition}\label{d-risk}
Let $\mathcal{X}$ be a set of risks that contains the constants. 

(a) A \textit{risk measure} is a functional $\rho:\mathcal{X}\to \R$.

(b) A risk measure $\rho$ is said to be \textit{coherent} if it satisfies the following conditions:

\begin{enumerate}[label=(\roman*)]
	\item If $X,Y\in \mathcal{X}$ with $X\leq Y$ then $\rho(X)\leq \rho(Y)$. (\textit{Monotonicity})
	\item If $X\in \mathcal{X}$ and $b\in \R$ with $X+b\in\mathcal{X}$ then $\rho(X+b)=\rho(X)+b$. (\textit{Cash-invariance})
	\item If $X\in \mathcal{X}$ and $\lambda\geq 0$ with $\lambda X\in\mathcal{X}$ then $\rho(\lambda X)=\lambda \rho(X)$. (\textit{Positive homogeneity})
	\item If $X,Y\in \mathcal{X}$ with $X+Y\in\mathcal{X}$ then $\rho(X+Y)\leq \rho(X)+\rho(Y)$. (\textit{Subadditivity})
\end{enumerate}
\end{definition}

The notion of coherence was introduced in \cite{ADEH99}. In the insurance literature, $\rho$ is also sometimes called a premium principle, see \cite{GMX20} or \cite{WaDh98}; see also Remark \ref{r-prime}(b) below. In the finance literature, the differing sign convention for risks, where positive values correspond to gains, leads to different notions of monotonicity and cash-invariance, see \cite{ADEH99}, \cite{FoSc16} or \cite{RoSg13}.

\begin{remark}\label{r-risk}
If $\mathcal{X}$ is a convex cone, then, for any $X,Y\in\mathcal{X}$, $b\in\R$, and $\lambda\geq 0$, $X+Y, X+b$, and $\lambda X\in\mathcal{X}$, so that the extra assumptions in (ii)--(iv) are not needed. But we will see in Example \ref{ex-cone} below that even for concave distortion functions the natural domain of definition of the corresponding distortion risk measure need not be a convex cone. 
\end{remark}

Another desirable property of risk measures is that they are \textit{law-invariant}: if a risk $X$ has the same distribution as a risk $Y\in \mathcal{X}$ then $X\in \mathcal{X}$ and $\rho(X)=\rho(Y)$. It will be immediately clear from their definitions that all the particular risk measures studied in this paper are law-invariant.

We next consider some continuity properties.

\begin{definition}\label{d-Fatou}
Let $\rho:\mathcal{X}\to \R$ be a risk measure.

(a) $\rho$ is said to have the \textit{Fatou property} if, for any sequence $(X_n)_n$ in $\mathcal{X}$ and $X,Y_1,Y_2\in\mathcal{X}$,
\begin{align*}
X_n\to X\ \&\ \forall n, Y_1\leq X_n \leq Y_2\Longrightarrow \rho(X)\leq \liminf_{n\to\infty} \rho(X_n).
\end{align*}

(b) $\rho$ is said to have the \textit{reverse Fatou property} if, for any sequence $(X_n)_n$ in $\mathcal{X}$ and $X,Y_1,Y_2\in\mathcal{X}$,
\begin{align*}
X_n\to X\ \&\ \forall n, Y_1\leq X_n \leq Y_2\Longrightarrow \rho(X)\geq \limsup_{n\to\infty} \rho(X_n).
\end{align*}

(c) $\rho$ is said to have the \textit{Lebesgue property} if, for any sequence $(X_n)_n$ in $\mathcal{X}$ and $X,Y_1,Y_2\in\mathcal{X}$,
\begin{align*}
X_n\to X\ \&\ \forall n, Y_1\leq X_n \leq Y_2\Longrightarrow \rho(X)=\lim_{n\to\infty} \rho(X_n).
\end{align*}
\end{definition}

Thus, $\rho$ has the Lebesgue property if and only if it has both the Fatou and the reverse Fatou property.

\begin{remark}\label{r-Fatou}
Some discussion of these definitions is in order.

(a) By a well known property, one can replace almost sure convergence by convergence in probability.

(b) In the literature, one usually demands that $|X_n|\leq Y$ for some $Y\in\mathcal{X}$. But this happens often in the context where $-Y\in\mathcal{X}$ whenever $Y\in\mathcal{X}$. In our context we found it useful to demand explicitly a lower bound from $\mathcal{X}$; see Example \ref{ex-Fatou}.

(c) Suppose that $\mathcal{X}$ has the property that, for any risk $X$, if there are $Y_1, Y_2\in \mathcal{X}$ with $Y_1\leq X\leq Y_2$ then $X\in\mathcal{X}$.

If $\rho$ is monotonic, then $\rho$ has the Fatou property if and only if, for any sequence $(X_n)_n$ in $\mathcal{X}$ and any $X\in \mathcal{X}$,
\[
X_n\nearrow X \Longrightarrow \rho(X_n)\to \rho(X);
\]
and $\rho$ has the reverse Fatou property if and only if, for any sequence $(X_n)_n$ in $\mathcal{X}$ and any $X\in \mathcal{X}$,
\[
X_n\searrow X \Longrightarrow \rho(X_n)\to \rho(X).
\]
This follows by passing to $\inf_{k\geq n} X_k$ and $\sup_{k\geq n} X_k$, respectively.

If $\rho$ is \textit{anti-monotonic}, that is, if $X,Y\in \mathcal{X}$ with $X\leq Y$ implies that $\rho(X)\geq \rho(Y)$, then, obviously, the arrows $\nearrow$ and $\searrow$ need to be interchanged; see also \cite[Section 4.2]{FoSc16}.

(d) The reverse Fatou property does not seem to have been given a name in the literature so far.

(e) By a remarkable result of Jouini, Schachermayer, and Touzi \cite{JST06}, see also \cite{Svi10} and \cite{LiMu25}, every law-invariant coherent risk measure on the space $L^\infty$ over an atom-less probability space has the Fatou property. For an extension to Orlicz spaces, see \cite[Corollary 2.5]{CGLL22}.
\end{remark}

\section{Distortion risk measures}\label{s-dist} 

\begin{definition}\label{d-distfunc}
A \textit{distortion function} is a function $g : [0, 1] \to [0, 1]$ that is increasing and right-continuous with $\lim_{u\nearrow 1}g(u)=g(1)=1$.
\end{definition}

In the literature, the requirements on a distortion function vary considerably. Often, $g(0)=0$ is also required; on this, see Example \ref{ex-class} below. Our choice is motivated by the well-known one-to-one correspondence between increasing and right-contin\-u\-ous functions $g : [0, 1] \to [0, 1]$ with $g(1)=1$ and Borel probability measures on $[0,1]$, which is given by $\mu_g([0,u])=g(u)$, $u\in [0,1]$. The Lebesgue-Stieltjes integral $\int_0^1 f \d g$ is then understood in the Lebesgue sense with respect to $\mu_g$. Note that we write $\int_0^1 f \d g$ instead of the more correct form $\int_{[0,1]} f \d g$, while $\int_{(0,1]} f \d g$ has possibly a different value. We also set $g(0-)=0$.

The distortion risk measures will be defined on the following space.

\begin{definition}\label{d-dist}
Let $g$ be a distortion function. Then $L_g=L_g(\Omega)$ is the space of all risks $X:\Omega\to\R$ such that
\[
\int_0^1 |F_X^{-1}(1-u)|\d g(u) <\infty.
\]
\end{definition}

\begin{remark} In \cite{Pic13}, Pichler seems to suggest that natural domains of risk measures have the property that if $X$ is a risk in the domain then so is $|X|$, see \cite[Proposition 5]{Pic13}. For example, if $g$ is given by $g(u)=\int_0^u w(v)\d v$, $u\in [0,1]$, then
Pichler takes as the natural domain of the distortion risk measure $\rho_g$ the set $\{X: \int_0^1 F_{|X|}^{-1}(1-u)w(u)\d u<\infty\}$, see \cite[Definition 8]{Pic13}.

The problem with this approach is that, by considering $|X|$, gains (corresponding to negative values) and losses (corresponding to positive values) are treated on the same footing. We therefore prefer to consider $|F_X^{-1}|$ instead of $F_{|X|}^{-1}$ in the above definition (and in Definition \ref{d-OL} below). 

We will continue the discussion in Remark \ref{rem-Lorentz}.
\end{remark}

Since, for any risk $X$, $F_X^{-1}(0)=-\infty$, $L_g$ would be empty if $\mu_g(\{1\})= g(1)-g(1-)>0$. This is the reason why we require $g(1-)=g(1)$ for our distortion functions. On the other hand, since $g(1-)=g(1)$, every bounded risk belongs to $L_g$, that is,
\[
L^\infty\subset L_g.
\]
In the same vein, if $g(0)>0$ then a risk $X$ can only belong to $L_g$ if $F_X^{-1}(1)<\infty$, which means that $X$ is bounded above.

\begin{definition}\label{d-distrm}
Let $g$ be a distortion function. The \textit{distortion risk measure} $\rho_g:L_g\to\R$ is given by
\[
\rho_g(X) = \int_0^1 F_X^{-1}(1-u)\d g(u).
\]
\end{definition}

We have useful alternative representations. For the first, compare for example with \cite[Section 2.6.1.2]{DDGK05} and \cite[Section 5.1]{DVGKTV06}, where, however, $g$ is the left-continuous version of ours; for the second, compare for example with \cite[p.\ 68]{Hei03} and \cite[Section 2.6.1.6]{DDGK05}. In this result, let
\[
g(u-)=\lim_{t\nearrow u} g(t),  \ u\in [0,1],
\]
be the left-hand limit; recall that $g(0-)=0$. 

\begin{proposition}\label{p-equiv} 
Let $X\in L_g$. 

\emph{(a)} We have
\begin{align*}
\rho&_g(X)= -\int_{-\infty}^0 (1-g(\overline{F}_X(x)-))\d x + \int_0^{\infty} g(\overline{F}_X(x)-)\d x,
\end{align*}
where $\overline{F}_X(x)=1-F_X(x)$.

\emph{(b)} Let $h:[0,1]\to [0,1]$ be given by $h(u)=1-g((1-u)-)$. Then
\[
\rho_g(X)= \int_{-\infty}^\infty x\, \d (h\circ F_X)(x).
\]
\end{proposition}

\begin{proof}
(a) Note that, by using Fubini and properties of $F_X^{-1}$, 
\begin{align*}
\rho_g(X) &= - \int_{F_X^{-1}(1-u)\leq 0} \int_{F_X^{-1}(1-u)\leq x \leq 0}\d x \d g(u)+\int_{F_X^{-1}(1-u)> 0} \int_{0\leq x< F_X^{-1}(1-u)}\d x \d g(u)\\
&=  - \int_{u\geq \overline{F}_X(0)} \int_{\overline{F}_X(x)\leq u,x\leq 0}\d x \d g(u)+\int_{u< \overline{F}_X(0)} \int_{\overline{F}_X(x)>u,x\geq 0}\d x \d g(u)\\
&=  -  \int_{x\leq 0} \int_{u\geq \overline{F}_X(x)}\d g(u) \d x +\int_{x\geq 0}\int_{u< \overline{F}_X(x)} \d g(u)\d x,
\end{align*}
which yields the claimed identity.

(b) The integral is the Lebesgue-Stieltjes integral with respect to the finite measure $\mu$ on $\R$ given by
\[
\mu((-\infty,x])=h(F_X(x)) = h(P(X\leq x));
\]
note that $h$ and $F_X$ are right-continuous and that $h(0)=0$, $\lim_{u\nearrow 1}h(u)=1-g(0)$, and $h(1)=1$. Writing $\mu_g$ for the probability measure on $[0,1]$ induced by $g$ we see that for $x\in\R$, using a property of $F_X^{-1}$,
\begin{align*}
\mu((-\infty,x]) &= 1 - g((1-F_X(x))-)\\
&= 1-\mu_g(u<1-F_X(x)) = \mu_g(u\geq 1-F_X(x))\\
&=\mu_g (F_X^{-1}(1-u)\leq x),
\end{align*}
so that  a push-forward measure argument and the definition of $\rho_g$ yield the claimed identity.
\end{proof}

In particular, for positive risks $X$, we find that
\begin{equation}\label{eq-defdist}
\rho_g(X) =\int_0^{\infty} g(\overline{F}_X(x)-)\d x = \int_0^\infty x\, \d (h\circ F_X)(x).
\end{equation}

\begin{remark}\label{r-percfct}
The formulas in the previous proposition show that the function $g$ distorts the probabilities $\overline{F}_X(x)=P(X>x)$, while the related function $h$ distorts the probabilities $F_X(x)=P(X\leq x)$. The functions $g$ and $h$ therefore represent the decision maker's perception of probabilities. For this reason, the distortion function $g$ is sometimes called a \textit{probability perception function}, see for example \cite{CCM97} or \cite{DDGKL06}. 
\end{remark}

\begin{example}\label{ex-class}
We have three classical examples of distortion risk measures. If $g(u)=u$ then $\rho_g(X)=E(X)$ on the set $L_g=L^1$ of integrable risks. If $g(u)=\mathds{1}_{[1-\alpha,1]}(u)$, $0 < \alpha < 1$, then $\rho_g(X)=\VaR_\alpha(X)=F_X^{-1}(\alpha)$ (Value at Risk) on the set of all risks; in the extreme case of $\alpha=0$ we have with $g(u)\equiv 1$ that $\rho_g(X)= \VaR_1(X)=\esssup X$ on the set of all risks for which $X^+\in L^\infty$; it therefore makes sense not to demand that $g(0)=0$. Finally, if $g(u) =\min\big(\frac{u}{1-\alpha},1\big)$, $0 < \alpha < 1$, then $\rho_g(X)=\TVaR_\alpha(X)=\frac{1}{1-\alpha}\int_\alpha^1 F_X^{-1}(u)\d u$ (Tail Value at Risk) on the set of all risks for which $X^+\in L^1$.
\end{example} 

We recall a well-known formula for TVaR, which is due to Rockafellar and Uryasev \cite{RoUr00}, \cite{RoUr02}, and Acerbi and Tasche \cite{AcTa02}; for short proofs, see \cite[p.\ 582]{DVGKTV06} or \cite[Proposition 4.51]{FoSc16}. It can be used, for example, to show that TVaR is subadditive, see \cite[Section 3.2]{EmWa15}. This type of formula will guide us throughout the paper, see Definitions \ref{d-HaezGoov} and \ref{d-dHG}.

\begin{proposition}\label{p-TVaR}
Let $0 < \alpha < 1$. If $X^+\in L^1$, then 
\[
\emph{TVaR}_\alpha(X) = \min_{x\in\mathbb{R}}\Big(\frac{1}{1-\alpha}E\big((X-x)^+\big)+x\Big),
\]
where the minimum is attained at $x=F_X^{-1}(\alpha)$.
\end{proposition}

The case of the Tail Value at Risk shows that $X\in L_g$ does not necessarily imply that $|X|\in L_g$. The following, however, is a direct consequence of the definition and the monotonicity of VaR.

\begin{proposition}\label{p-lg}
If $Y_1,Y_2\in L_g$ and $Y_1\leq X\leq Y_2$ then $X\in L_g$.
\end{proposition}

The next example shows a rather unexpected problem with the domain of distortion risk measures, which does not seem to have been noticed before.

\begin{example}\label{ex-cone}
There exists a concave distortion function $g$ for which $L_g$ is not a convex cone. Indeed, consider $g:[0,1]\to[0,1]$ given by $g(u)=\frac{4}{3}(1-\e^{-3})u\mathds{1}_{[0,\frac{3}{4})}(u)+(1-\e^{-\frac{u}{1-u}})\mathds{1}_{[\frac{3}{4},1)}(u)$ with $g(1)=1$.

On $\Omega=[-1,1]$ with the normalized Lebesgue measure, we consider $X(\omega)=-\e^{\frac{1}{|\omega|}}\mathds{1}_{[-1,0)}(\omega)$ and $Y(\omega)=-\e^{\frac{1}{|\omega|}}\mathds{1}_{(0,1]}(\omega)$. We calculate that $F_X(x)=F_Y(x)= \frac{1}{2\ln(-x)}\mathds{1}_{(-\infty,-\e)}(x) + \frac{1}{2}\mathds{1}_{[-\e,0)}(x) +\mathds{1}_{[0,\infty)}(x)$ for $x\in\R$ and $F_X^{-1}(u)=-\e^{\frac{1}{2u}}\mathds{1}_{(0,\frac{1}{2}]}(u)$ for $u\in(0,1]$. Then $\int_0^1 |F_X^{-1}(1-u)|\d g(u)=\int_{\frac{1}{2}}^1 \e^{\frac{1}{2(1-u)}} g'(u)\d u= C + \int_{\frac{3}{4}}^1 \e^{\frac{1}{2(1-u)}} \frac{1}{(1-u)^2}\e^{-\frac{1}{1-u}}\e\,\d u<\infty$, where $C$ is some constant, so that $X\in L_g$ and hence also $Y\in L_g$. On the other hand, $F_{X+Y}(x)= \frac{1}{\ln(-x)}\mathds{1}_{(-\infty,-\e)}(x) + \mathds{1}_{[-\e,\infty)}(x)$ and $F_{X+Y}^{-1}(u)= -\e^{\frac{1}{u}}$, so that $\int_0^1 |F_{X+Y}^{-1}(1-u)|\d g(u)\geq \int_{\frac{3}{4}}^1  \e^{\frac{1}{1-u}} \frac{1}{(1-u)^2}\e^{-\frac{1}{1-u}}\e\,\d u =\infty$, which shows that $X+Y\notin L_g$. Thus $L_g$ is not a convex cone.
\end{example}

In Subsection \ref{subs-convex} we will find conditions on $g$ under which $L_g$ is a convex cone.

The following is an immediate consequence of the corresponding properties for the Value at Risk.

\begin{proposition}\label{p-distmonposhom}
The distortion risk measure $\rho_g$ is monotonic, cash-invariant and positively homogeneous on $L_g$.
\end{proposition}

It is well known that while Value at Risk is not subadditive, it is (even) additive for comonotonic risks, see \cite[Theorem 4.2.1]{DVGKTV06}. There are various ways to define comonotonicity, see \cite[Definition 4, Theorem 2]{DDGKV02}. Maybe the one that expresses best the idea behind this notion is to say that two risks $X$ and $Y$ are \textit{comonotonic} if there is a random variable $Z$ with values in an interval $I\subset \R$ and two increasing functions $f_1,f_2:I\to\R$ such that $(X,Y)$ and $(f_1(Z),f_2(Z))$ have the same distribution.

Now, the definition of the distortion risk measures and the mentioned property of $\VaR$ immediately imply the following; see also \cite[p.\ 593]{DVGKTV06}. 

\begin{proposition}\label{p-distcomon}
Let $X,Y\in L_g$ be comonotonic. Then $X+Y\in L_g$ and    
\[
\rho_g(X+Y)=\rho_g(X)+\rho_g(Y).
\]
\end{proposition}

The next result is well known if $g(0)=0$, see \cite{DVGKTV06}, \cite{Wa96}, \cite{WaDh98}. In general, $g$ is a convex combination of the constant distortion function $g_1=1$ and a distortion function $g_2$ with $g_2(0)=0$. Thus $\rho_g$ is a convex combination of $\rho_{g_1}=\VaR_1=\esssup$ and $\rho_{g_2}$, and both are coherent.

\begin{theorem} 
If $g$ is concave, then the distortion risk measure $\rho_g$ is coherent on $L_g$.
\end{theorem}

We will give a proof of the theorem for the more general distortion Haezendonck-Goovaerts risk measures in Section \ref{s-distHaezGoov}.

We next turn to continuity properties.

\begin{proposition}\label{p-distFatou}
The distortion risk measure $\rho_g$ has the Fatou property on $L_g$.
\end{proposition}

\begin{proof}
By Remark \ref{r-Fatou}(c) and Propositions \ref{p-lg} and \ref{p-distmonposhom} it suffices to show that if $X_n\nearrow X$  and $X_1, X\in L_g$ then $\rho_g(X_n)\to\rho_g(X)$. Now, the hypothesis implies that $\overline{F}_{X_n}(x)\nearrow \overline{F}_{X}(x)$ for all $x\in\mathbb{R}$ with at most countably many exceptions. Since $u\mapsto g(u-)$ is left-continuous and increasing, we deduce that $g(\overline{F}_{X_n}(x)-)\nearrow g(\overline{F}_{X}(x)-)$ for these $x$. Since $X_1\leq X_n \leq X$ for all $n$, Proposition \ref{p-equiv}(a) and the dominated convergence theorem implies that $\rho_g(X_n)\to\rho_g(X)$. 
\end{proof}

In particular, we have the following, which should be known, but we haven't been able to find a reference.

\begin{corollary} 
For $0<\alpha\leq 1$, the Value at Risk \emph{VaR$_\alpha$} has the Fatou property.
\end{corollary}

\begin{example}\label{ex-Fatou}
On $\Omega=[0,1]$ with the Lebesgue measure, we consider the risks $X_n=-n\mathds{1}_{[0,1/n]}$, $n\geq 2$, so that $X_n\to X:=0$. Let $g$ be the distortion function $g(u)=(-1+2u)\mathds{1}_{[1/2,1]}(u)$. Then $\rho_g(X_n)= \int_{1-1/n}(-n)2\d u=-2$, and hence $\rho_g(X)>\liminf_{n\to\infty}\rho_g(X_n)$. On the other hand, taking $Y=\sup_{n\geq 2}|X_n|$, one verifies that $Y\in L_g$; note however that $-Y\notin L_g$. This example shows that while the Fatou property holds on $L_g$ in the sense of Definition \ref{d-Fatou}, it would not hold if we had only demanded that $\forall n, |X_n|\leq Y$ for some $Y\in L_g$. 
\end{example}

We turn to the reverse Fatou property.

\begin{proposition}\label{p-distrevFatou}~

\emph{(a)} If $g(0)=0$ and $g$ is continuous, then $\rho_g$ has the reverse Fatou property on $L_g$, and hence the Lebesgue property.

\emph{(b)} If the underlying probability space $(\Omega,\mathcal{A},P)$ is atomless, then $\rho_g$ has the reverse Fatou property on $L_g$ if and only if $g(0)=0$ and $g$ is continuous.
\end{proposition}

\begin{proof}
(a) This follows exactly as in the proof of Proposition \ref{p-distFatou}, taking account of the continuity of $g$; note that $g(0-)=0$. 

(b) Suppose that $g$ is not continuous or that $g(0)\neq 0$. Then there is some $u\in\mathopen[0,1)$ such that $g(u-)<g(u)$. Let $(p_n)_n$ be a strictly decreasing sequence in $[0,1]$ with limit $u$. If $P$ is atomless, there exists a decreasing sequence $(A_n)_n$ of sets in $\mathcal{A}$ with $P(A_n)=p_n$, $n\geq 1$, see \cite[Theorem 8.14.2]{GeGe25}. Then $A:=\bigcap_{n=1}^\infty A_n$ satisfies $P(A)=u$. Let $X_n=\mathds{1}_{A_n}$ and $X=\mathds{1}_{A}$, which belong to $L_g$ as bounded risks. Then $X_n\to X$ on $\Omega$ and $0\leq X_n\leq 1$ for all $n$, with $0,1\in L_g$. Moreover, by \eqref{eq-defdist}, we find that $\rho_g(X_n)=g(p_n-)\geq g(u)$ for all $n$ and $\rho_g(X)=g(u-)$, so that $\rho_g(X)<\limsup_{n\to\infty}\rho_g(X_n)$, contradicting the reverse Fatou property.
\end{proof}

\begin{remark}
Since the counter-example is taken from $L^\infty$, the proposition remains true in the more restrictive setting of $L^\infty$.
\end{remark}

We finally discuss an interesting link between the domain $L_g$ and Lorentz spaces.

\begin{remark}\label{rem-Lorentz}
In functional analysis, 
\[
X^*(u)=F_{|X|}^{-1}(1-u), \ u\in [0,1),
\]
with $X^*(1)=0$, is known as the nonincreasing rearrangement of $X$, see \cite{CRS07}, \cite{GeGe25}, \cite{PKJF13}. If $w:[0,1]\to \mathbb{R}$ is a positive measurable function with $\int_0^1 w(u)\d u=1$, then
\[
\Lambda(w)= \Big\{X : \|X\|:=\int_0^1 X^\ast(u) w(u) \d u <\infty\Big \}
\]
is called a (classical) Lorentz space, see \cite{CRS07}, \cite{Lor51}, \cite{PKJF13}. Setting $g(u)=\int_0^u w(v)\d v$, $u\in [0,1]$, we obtain a continuous distortion function with $g(0)=0$. Then $\Lambda(w)= \{X: |X|\in L_g\}$ and $\|X\|=\rho_g(|X|)$ for $X\in \Lambda(w)$.

For these distortion functions $g$, one can even define $L_g$ and $\rho_g$ completely in terms of notions introduced by Lorentz. Indeed, $X\in L_g$ if and only if $X^+\in \Lambda(w)$ and $\inf_{x\in\mathbb{R}} (\|(X-x)^+\|+x)>-\infty$; in that case, the infimum gives $\rho_g(X)$. For the proof see Proposition \ref{p-LorentzLgRhog} in the Appendix. 

Now, decreasing functions $w$ correspond to concave distortion functions $g$ with $g(0)=0$. In that case, and for $\Omega=[0,1]$, Lorentz \cite{Lor51} showed that $\|\cdot\|$ defines a norm on $\Lambda(w)$; for general spaces $\Omega$, see \cite[Theorem 2.5.1]{CRS07}. The fact that $\|\cdot\|$ is a norm implies that the corresponding distortion risk measure is subadditive on the positive cone of $L_g$. In addition, one can show that $\Lambda(w)\subset L_g$; we give the proof in the Appendix, see Proposition \ref{p-LorentzLg}.
\end{remark}

The connection between distortion risk measures and Lorentz space norms was also recently noted in \cite[Section 4.5]{FrWi24}.

\section{Haezendonck-Goovaerts risk measures}\label{s-HaezGoov} 

We recall here the definition of the Haezendonck-Goovaerts risk measures. They are defined on Orlicz spaces, which are well-known spaces from functional analysis, see \cite{Che96}, \cite[Chapter 2]{EdSo92}, \cite{PKJF13} or \cite{RaRe91}.

\begin{definition}
A \textit{Young function} is a convex function $\phi : [0, \infty) \to [0, \infty)$ with $\phi(0)=0$ and $\lim_{t\to\infty}\phi(t)=\infty$. The corresponding \textit{Orlicz space} $L^\phi=L^\phi(\Omega)$ is the space of all risks $X:\Omega\to\R$ for which there is some $a>0$ such that
\[
E\Big(\phi\Big(\frac{|X|}{a}\Big)\Big) <\infty.
\]
\end{definition}

Young functions are also known as Orlicz functions. They are sometimes assumed to be strictly increasing (see \cite{BeRo08}), and they are often assumed to be \textit{normalized}, that is, $\phi(1)=1$ (see \cite{BeRo08}, \cite{BeRo12}, \cite{HaGo82}). If $\phi(1)>0$, normalization can always be achieved by replacing $\phi$ with $\frac{\phi}{\phi(1)}$. 

We have that
\[
L^\infty\subset L^\phi\subset L^1,
\]
see Proposition \ref{p-incl} below for a generalization. Moreover, $L^\phi$ is a vector space, see \cite[Theorem 2.1.11]{EdSo92}. 

As a preliminary step towards the Haezendonck-Goovaerts risk measures, one considers the Orlicz premia, which are defined for positive risks. We denote by $L^\phi_+$ the convex cone of positive risks in $L^\phi$.

\begin{definition}\label{d-Orlicz}
Let $\phi$ be a Young function and $\alpha< 1$. The \textit{Orlicz premium} $\pi_{\phi,\alpha}:L^\phi_+\to\R$ is given by
\[
\pi_{\phi,\alpha}(X) = \inf\Big\{a>0 : E\Big(\phi\Big(\frac{X}{a}\Big)\Big)\leq 1-\alpha\Big\}.
\]
\end{definition}

For $\alpha=0$, the Orlicz premium coincides with the Luxemburg norm in the Orlicz space $L^\phi$, see \cite{Che96}, \cite{EdSo92}.

\begin{remark}\label{r-prime}
(a) We interpret $\phi(X)$ as the evaluation of the risk $X$ by the decision maker. Since the role of a risk measure (and of a premium, see below) is to be on the prudent side, the value of $\phi(X)$ should be proportionally larger for larger values of $X$, meaning that $\phi$ is not only increasing but convex. Now let us extend $\phi$ in an increasing and convex way to all of $\mathbb{R}$. Since, by our sign convention, the financial position associated to the risk $X$ is $-X$, it makes sense to write $\phi(X) = -U(-X)$, where $U$ is an increasing concave function, that is, a (risk averse) \textit{utility function}. In other words, the function $\phi$ is, up to a sign change, a utility function.                        

(b) We add a word  on terminology. In financial mathematics, risk measures are meant to quantify the ``downside risk'' (\cite[p.\ 194]{FoSc16}) or the ``riskiness'' of a risk (\cite[p.\ 61]{DDGK05}). This is well captured by VaR and its variants. Thus, while the expectation $E(X)$ is a coherent risk measure, it is of little interest in this context. In insurance mathematics, the insurance risk is the ``amount of money paid by an insurance company to indemnify a policyholder'' (\cite[Definition 1.4.3]{DDGK05}). In return, the insurer receives a premium. This is well captured by the expectation $E(X)$, called the net premium (\cite[p.\ 61]{DDGK05}), and its variants. The Orlicz space norm being such a variant, it seems more appropriate to call $\pi_{\phi,\alpha}$ a premium (as, for example, in \cite{HaGo82} and \cite{BeRo08}) than a risk measure. But see Remark \ref{r-return} below.
\end{remark} 

We note that, while Haezendonck and Goovaerts \cite{HaGo82} only consider $\alpha=0$, later work requires that $\alpha\in [0,1)$, see \cite{BeRo08} and \cite{BeRo12}, in each case with a normalized $\phi$. 

We have that $\pi_{\phi,\alpha}$ takes finite values because $E\big(\phi\big(\frac{X}{a}\big)\big)\to 0$ as $a\to\infty$ by the dominated convergence theorem.

\begin{remark}\label{rem-Orlicz}
If $X\neq 0$ then the infimum in the definition of $\pi_{\phi,\alpha}(X)$ is attained. If, moreover, $X\in L^\infty_+$, or else $X\in L^\phi_+$ and $\phi$ satisfies the $\Delta^2$-condition, see \eqref{eq-delta2} below, then there is a unique value $a>0$ such that 
$E\big(\phi\big(\frac{X}{a}\big)\big)= 1-\alpha$; and $a=\pi_{\phi,\alpha}(X)$. These facts are given in \cite[Theorem 2]{HaGo82} and \cite[p.\ 108]{BeRo12}. A proof in a more general situation will be given in Proposition \ref{p-distOrlicz} below. Note also that, by an example given in \cite[pp. 45-46]{HaGo82}, one cannot drop the $\Delta^2$-condition in the statement above.
\end{remark}

We collect the main properties of the Orlicz premia. For normalized $\phi$ and bounded risks, the first two results were obtained in \cite[Theorem 2]{HaGo82} if $\alpha=0$ and stated in \cite[Proposition 2]{BeRo08} if $\alpha\geq 0$. 
                                                               
\begin{proposition}\label{p-Orlicz}~
 
\begin{enumerate}[label=\emph{(\alph*)}]
\item For any $X\in L_+^\phi$,
\[
\frac{E(X)}{\phi^{-1}(1-\alpha)}\leq\pi_{\phi,\alpha}(X)\leq \frac{\esssup X}{\phi^{-1}(1-\alpha)}.
\]
\item For any $b\geq 0$, $\pi_{\phi,\alpha}(b)= \frac{b}{\phi^{-1}(1-\alpha)}$.
\end{enumerate}
\end{proposition}

\begin{theorem}\label{t-Orlicz}
The Orlicz premium $\pi_{\phi,\alpha}$ is monotonic, positively homogeneous, and subadditive on $L_+^\phi$.
\end{theorem}

The next two results were recently obtained, for $0<\alpha<1$, inside the proofs of \cite[Theorems 3.3 and 3.4]{GMX20}; see also \cite[Proposition 2]{BeRo08}.

\begin{proposition} 
The Orlicz premium $\pi_{\phi,\alpha}$ has the Fatou property on $L_+^\phi$.
\end{proposition}

Recall that a Young function satisfies the $\Delta_2$-condition if there exist $s\geq 0$ and $K>0$ such that
\begin{equation}\label{eq-delta2}
\phi(2t)\leq K\phi(t)
\end{equation}
for all $t\in [s,\infty)$. 

\begin{proposition}\label{p-OrevFatou}
\emph{(a)} If $\phi$ satisfies the $\Delta_2$-condition, then $\pi_{\phi,\alpha}$ has the reverse Fatou property on $L_+^\phi$, and hence the Lebesgue property.

\emph{(b)} If the underlying probability space $(\Omega,\mathcal{A},P)$ is atomless, then $\pi_{\phi,\alpha}$ has the reverse Fatou property on $L_+^\phi$ if and only if $\phi$ satisfies the $\Delta_2$-condition.
\end{proposition}

On $L^\infty_+$, there is no restriction, see \cite[Proposition 17]{BeRo08} when $\alpha\geq 0$.

\begin{proposition} 
The Orlicz premium $\pi_{\phi,\alpha}$ has the reverse Fatou property on $L^\infty_+$.
\end{proposition}

These five results will be proved, in greater generality, in Propositions \ref{p-distOrlicz2}, \ref{p-distOrliczbis}, \ref{p-distOrliczFatou}, \ref{p-OLrevFatou}, \ref{p-OLrevFatoubdd} and Theorem \ref{t-distOrlicz} below.

\begin{remark}\label{r-return}
Bellini, Laeven, and Rosazza Gianin \cite{BLR18} have introduced an interesting new notion in risk theory. A functional $\widetilde{\rho}:L^\infty_+\to [0,\infty)$ is called a \textit{return risk measure} if it is monotonic, positively homogeneous and normalized (in the sense that $\widetilde{\rho}(1)=1$). Thus, for a normalized Young function $\phi$ and $\alpha=0$, the Orlicz premium $\pi_{\phi,0}$ is a return risk measure.

Return risk measures induce monotonic, cash-invariant risk measures $\rho:L^\infty\to\R$ that satisfy $\rho(0)=0$ (also called \textit{monetary risk measures}, see \cite[Definition 4.1]{FoSc16}) by setting $\rho(X)= \log(\widetilde{\rho}(\exp X))$; see \cite[Section 3]{BLR18} for the relationship between $\rho$ and $\widetilde{\rho}$. This leads to the following interpretation of return risk measures for decision makers: Given a capital $k>0$, the log-return of a financial position $X$ relative to $k$, that is, $\log(\tfrac{X}{k})$ is acceptable (with respect to $\rho$) if and only if $\widetilde{\rho}(\tfrac{X}{k})\leq 1$; see also \cite[p.\ 12]{BLR18}. Thus, return risk measures provide a relative assessment of risk.

Return risk measures have recently been studied intensively, see, for example, \cite{ABL23}, \cite{ABL24}, \cite{ABL25}, \cite{LRZ24} (which considers them on larger classes $\mathcal{X}_+$ of positive risks), \cite{LRZ26}, and \cite{LLS25}. In our context it is interesting to note that, for $\alpha=0$, Orlicz premia define return risk measures for a wide class of functions $\phi$ that are not necessarily convex, see Ayg\"un, Bellini, and Laeven \cite[Proposition 8(a)]{ABL25}.
\end{remark}

Returning to the classical notion of (monetary) risk measures, Orlicz premia in general lack the important property of cash-invariance. It suffices to consider $\phi(t)=t^2$ and $\alpha=0$, so that $\pi_{\phi,0}(X) = E(X^2)^{\frac{1}{2}}$. It is surprising that a simple procedure allows to add cash-invariance while preserving the other three properties of coherence. 

\begin{definition}\label{d-HaezGoov}
Let $\phi$ be a normalized Young function and $\alpha\in [0,1)$. The \textit{Haezendonck-Goovaerts risk measure} $\rho_{\phi,\alpha}:L^\phi\to\R$ is given by
\begin{equation}\label{eq-HG}
\rho_{\phi,\alpha}(X) = \inf_{x\in\R} \big(\pi_{\phi,\alpha}((X-x)^+)+x\big).
\end{equation}
\end{definition}

In this definition, we have restricted $\phi$ and $\alpha$. The minimal requirement would be that $\alpha \geq 1 -\phi(1)$. Indeed, if $\alpha < 1 -\phi(1)$, that is, $\phi^{-1}(1-\alpha)>1$, then it follows from Proposition \ref{p-Orlicz}(b) that, for $x\leq 1$, $\pi_{\phi,\alpha}((1-x)^+)+x = \frac{1-x}{\phi^{-1}(1-\alpha)}+x = \frac{1+x(\phi^{-1}(1-\alpha)-1)}{\phi^{-1}(1-\alpha)}$, so that $\rho_{\phi,\alpha}(1)=-\infty$. Thus, in order to have a risk measure, we need to impose that $\alpha \geq 1 -\phi(1)$. Since also $\alpha<1$, $\phi(1)$ must be nonzero. Hence we can normalize $\phi$, and then $\alpha\in [0,1)$. 

Now, whenever $X\in L^\phi$, then $(X-x)^+\in L_+^\phi$ for any $x\in \R$; also, under the assumptions on $\phi$ and $\alpha$, the infimum is in $\R$. We will show these assertions in more generality in Remark \ref{r-distHaezGoov} below. Thus, $\rho_{\phi,\alpha}$ is a well-defined risk measure.  

\begin{remark}
The Haezendonck-Goovaerts risk measures were introduced by Goovaerts, Kaas, Dhaene, and Tang \cite{GKDT04} in a slightly different form. Formula \eqref{eq-HG} is due to Bellini and Rosazza Gianin \cite{BeRo08}, who were motivated by the representation of TVaR given in Proposition \ref{p-TVaR}; see also \cite[p.\ 989]{BeRo08} for a nice discussion.
\end{remark}

\begin{example}\label{ex-HaezGoov}
Let us take for $\phi$ the identity function. Then, for $X\in L^\phi_+=L^1_+$, $\pi_{\phi,\alpha}(X) = \frac{1}{1-\alpha}E(X)$, $\alpha<1$. Thus, if $0<\alpha<1$ and $X\in L^\phi=L^1$, then $\rho_{\phi,\alpha}(X)= \text{TVaR}_\alpha(X)$ by Proposition \ref{p-TVaR}.
\end{example}

It follows easily from Theorem \ref{t-Orlicz} that the function $x\mapsto \pi_{\phi,\alpha}((X-x)^+)+x$ is convex for any $\alpha<1$, see also \cite[Proposition 3(a)]{BeRo12}. Moreover, for $0<\alpha<1$, it was shown in \cite[Proposition 3(b)]{BeRo12} that the function has a minimum, that is, the infimum in \eqref{eq-HG} is attained. We will prove more general results in Propositions \ref{p-distHaezGoovconv}(a) and \ref{p-infmin}(a) below. We will also see in Proposition \ref{p-infmin}(b) that the minimum is unique if $\phi$ is strictly convex and satisfies the $\Delta_2$-condition and if $P(X=\esssup X)=0$.

It turns out that the case of $\alpha=0$ is exceptional. In \cite[Example 15]{BeRo08}, an example was given where the infimum in \eqref{eq-HG} is not attained. This led subsequent authors to only consider the case of $\alpha>0$; see for example \cite[p. 79]{AhSh14}. We will first show that the example in \cite{BeRo08} is, in fact, a special case of a very general situation. 

\begin{proposition}\label{pex-HaezGoov}
Let $\alpha=0$ and $X\in L^\phi$.

\emph{(a)} Then $x\mapsto \pi_{\phi,0}((X-x)^+)+x$ is increasing on $\mathbb{R}$. In particular,
\[
\rho_{\phi,0}(X) = \lim_{x\to-\infty} \big(\pi_{\phi,0}((X-x)^+)+x\big).
\]

\emph{(b)} Let $\phi$ be strictly convex and satisfy the $\Delta_2$-condition. If $P(X=\esssup X)=0$ then $x\mapsto \pi_{\phi,0}((X-x)^+)+x$ is strictly increasing on $\mathbb{R}$. In particular, the function does not attain its infimum.
\end{proposition}

A more general result will be proved in Proposition \ref{p-distHaezGoovinf} below.

We now collect the main properties of Haezendonck-Goovaerts risk measures; the results were obtained for certain subspaces of $L^\phi$ in \cite[Proposition 12]{BeRo08} and \cite[Theorems 3.1, 3.2]{GKDT04}. The general case will follow from Proposition \ref{p-distHaezGoov} and Theorem \ref{t-distHaezGoov2} below. In these results, when we do not restrict $\alpha$, we suppose that $\alpha\in[0,1)$. 

\begin{proposition}\label{p-HaezGoov}
Let $X\in L^\phi$. Then:
\begin{enumerate}[label=\emph{(\alph*)}] 
\item $\rho_{\phi,\alpha}(X) \leq \pi_{\phi,\alpha}(X^+)$.
\item $E(X)\leq \rho_{\phi,\alpha}(X) \leq \esssup X$.
\item If $\alpha\neq 0$ then $\rho_{\phi,\alpha}(X) \geq \VaR_\alpha(X)$.
\end{enumerate}
\end{proposition}

\begin{theorem}\label{t-HaezGoov}
The Haezendonck-Goovaerts risk measure $\rho_{\phi,\alpha}$ is coherent on $L^\phi$.
\end{theorem}

As for continuity properties, the following were obtained in \cite[Theorems 3.3 and 3.4]{GMX20} and \cite[Proposition 17]{BeRo08} for $0<\alpha<1$. The case of $\alpha=0$ is of little interest, see Theorem \ref{t-HGtriv} below; but see also Problem \ref{pr-Fatou0}.

\begin{proposition}\label{p-HGFatou}
If $0<\alpha<1$, then $\rho_{\phi,\alpha}$ has the Fatou property on $L^\phi$.
\end{proposition}

\begin{proposition}~

\emph{(a)} If $\phi$ satisfies the $\Delta_2$-condition then $\rho_{\phi,\alpha}$ has the reverse Fatou property on $L^\phi$.

\emph{(b)} Let $0<\alpha<1$. If the underlying probability space $(\Omega,\mathcal{A},P)$ is atomless and if $\rho_{g,\phi,\alpha}$ has the reverse Fatou property on $L^\phi$ then $\phi$ satisfies the $\Delta_2$-condition.
\end{proposition}

\begin{proposition} 
The Haezendonck-Goovaerts risk measure $\rho_{\phi,\alpha}$ has the reverse Fatou property on $L^\infty$.
\end{proposition}
 
These results will be generalized in Propositions \ref{p-distHGFatou}, \ref{p-distHGFatourev}, \ref{p-distHGFatourevbdd}, and \ref{p-distHGFatourevnec} below. 

Finally, as we have seen, the case $\alpha=0$ is quite exceptional. Indeed, in that case, the Haezendonck-Goovaerts risk measure is trivial on bounded risks, in some sense. This surprising fact does not seem to have been observed before.

\begin{theorem}\label{t-HGtriv}
Let $\alpha=0$. Then, for all $X\in L^\infty$,
\begin{equation}\label{eq-alpha0}
E(X)\leq \rho_{\phi,0}(X)\leq \frac{c_+}{c_-}E(X^+) - \frac{c_-}{c_+}E(X^-),
\end{equation}
where $c_-$ is the left derivative of $\phi$ at $1$, and $c_+$ is the right derivative of $\phi$ at $1$. If $\phi$ satisfies the $\Delta_2$-condition then this holds for all $X\in L^\phi$.
\end{theorem}

\begin{corollary}\label{c-HGexp}
Let $\alpha=0$. If $\phi$ is differentiable at $1$ and satisfies the $\Delta_2$-condition, then, for all $X\in L^\phi$,
\[
\rho_{\phi,0}(X)=E(X).
\]
\end{corollary}

For example, for the natural choice of $\phi(t)= t^p$, $p\geq 1$, $\rho_{\phi,0}$ coincides with the expectation, which is not considered a good risk measure. 

We will obtain a more general result below, see Theorem \ref{t-dHGtriv} with Corollary \ref{c-dHGtriv}.

The following example shows that, if $\phi$ is not differentiable at 1, $\rho_{\phi,0}$ need not reduce to the expectation.

\begin{example}
We consider the normalized Young function $\phi(t)=t$, $0\leq t\leq 1$, and $\phi(t)=2t-1$, $t>1$. Let $X$ be uniformly distributed on $[0,1]$. One calculates that, for $x<0$, $\pi_{\phi,0}((X-x)^+)+x= 2-\sqrt{2}$, so that, by Proposition \ref{pex-HaezGoov}(a), $\rho_{\phi,0}(X)=2-\sqrt{2}>E(X)$. Also, $X\geq 0$ and $\rho_{\phi,0}(X)\leq 2E(X)$, confirming \eqref{eq-alpha0}.

Moreover, if we take $X$ to be uniformly distributed on $[-1,0]$, then, by cash-invariance, $\rho_{\phi,0}(X)=1-\sqrt{2}>E(X)$. Also, $X\leq 0$ and $\rho_{\phi,0}(X)\leq \tfrac{1}{2}E(X)$, confirming again \eqref{eq-alpha0}.
\end{example}

Interestingly, the Haezendonck-Goovaerts risk measures are of little interest as \textit{monetary} risk measures precisely when the 
corresponding Orlicz premia become \textit{return} risk measures, that is, when $\alpha=0$ (see Remark \ref{r-return}).

\section{Distortion Haezendonck-Goovaerts risk measures}\label{s-distHaezGoov} 

We now come to the main contribution of this work: the combination of distortion risk measures and Haezendonck-Goovaerts risk measures into a single new class of risk measures. This was suggested in 2012 by Goovaerts, Linders, Van Weert, and Tank \cite[Definition 4.2]{GLVT12}.

\subsection{The domain}
We begin by defining the set of risks where the distortion Haezendonck-Goovaerts risk measures will be defined.

By a property of quantile functions we have that
\[
E\Big(\phi\Big(\frac{|X|}{a}\Big)\Big) =  \int_0^1 \phi\Big(\frac{|F_X^{-1}(1-u)|}{a}\Big) \d u.
\]
Motivated by this we are led to distort $L^\phi$ into a new space $L^\phi_g$; we refrain from giving this space a name, see Remark \ref{rem-OL}(a).

\begin{definition}\label{d-OL}
Let $g$ be a distortion function and $\phi$ a Young function. Then $L^\phi_g=L^\phi_g(\Omega)$ is the space of all risks $X:\Omega\to\R$ for which there is some $a>0$ such that
\[
\int_0^1 \phi\Big(\frac{|F_X^{-1}(1-u)|}{a}\Big) \d g(u) <\infty.
\]
\end{definition}

By the above, if $g$ is the identity then $L_g^\phi=L^\phi$; and if $\phi$ is the identity then $L_g^\phi=L_g$.

As in our discussion in Section \ref{s-dist} we see that if $g(0)>0$ then $X\in L_g^\phi$ implies that $X$ is bounded above. And the fact that $g(1-)=g(1)$ implies that the bounded risks belong to $L^\phi_g$. Indeed, we have the following.

\begin{proposition}\label{p-incl}
We have that
\[
L^\infty \subset L^\phi_g \subset L_g.
\]
\end{proposition}

\begin{proof} For the second inclusion, note that since $\phi$ is convex and necessarily increasing there are $c>0$ and $b\in \R$ such that $\phi(t)\geq ct+b$ for all $t\geq 0$. Thus
\begin{align*}
\int_0^1 \phi\Big(\frac{|F_X^{-1}(1-u)|}{a}\Big) \d g(u)\geq \frac{c}{a}\int_0^1|F_X^{-1}(1-u)| \d g(u) +b, 
\end{align*}
so that $X\in L_g^\phi$ implies that $X\in L_g$.
\end{proof} 

We have seen in Example \ref{ex-cone} that $L_g^\phi$ is not necessarily a convex cone, even if $g$ is concave and $\phi$ is the identity. In Subsection \ref{subs-convex} we will present conditions on a concave distribution function so that $L_g^\phi$ is a convex cone, for any Young function. 

Also, by Section \ref{s-dist}, $X\in L_g^\phi$ does not necessarily imply that $|X|\in L_g^\phi$. Instead, the definition implies the following.

\begin{proposition}\label{p-inv}
If $Y_1,Y_2\in L_g^\phi$ and $Y_1\leq X\leq Y_2$ then $X\in L_g^\phi$.
\end{proposition}

\subsection{Orlicz-Lorentz premia}

We start the definition of the distortion Haezendonck-Goovaerts risk measures by distorting the Orlicz premia.

We denote by $(L_g^\phi)_+$ the set of positive risks in $L_g^\phi$. Since $F_X^{-1}$ is positive for such risks we have that $X\in (L_g^\phi)_+$ if and only if $X\geq 0$ and
\[
\int_0^1 \phi\Big(\frac{F_X^{-1}(1-u)}{a}\Big) \d g(u)<\infty
\]
for some $a>0$.

\begin{definition}
Let $g$ be a distortion function, $\phi$ a Young function, and $\alpha< 1$. The \textit{Orlicz-Lorentz premium} $\pi_{g,\phi,\alpha}:(L^\phi_g)_+\to\R$ is given by
\begin{align*}
\pi_{g,\phi,\alpha}(X)= \inf\Big\{a>0 : \int_0^1 \phi\Big(\frac{F_X^{-1}(1-u)}{a}\Big) \d g(u)\leq 1-\alpha\Big\}.
\end{align*}
\end{definition}

If $g$ is the identity then the Orlicz-Lorentz premium $\pi_{g,\phi,\alpha}=\pi_{\phi,\alpha}$ is the Orlicz premium; and if $\phi$ is the identity then $\pi_{g,\phi,0}=\rho_g$ is the distortion risk measure (on the positive risks in $L_g$).

We see as in Section \ref{s-HaezGoov} that $\pi_{g,\phi,\alpha}$ takes finite values. 

\begin{remark}\label{rem-OL}
(a) The premium is named after the Orlicz-Lorentz spaces of functional analysis, see \cite{HKM02}, \cite{Kam90}, \cite[Section 5]{KMP03}. If $w:[0,1]\to\mathbb{R}$ is a positive measurable function with $\int_0^1w(u)\d u=1$ and $\phi$ is a Young function, then the Orlicz-Lorentz space $\Lambda_{\phi,w}=\Lambda_{\phi,w}(\Omega)$ is defined as the space of all measurable functions $X$ on $\Omega$ such that 
\[
\int_0^1 \phi\Big(\frac{X^*}{a}\Big)w(u) \d u<\infty \text{ for some $a>0$},
\] 
where $X^\ast$ is the nonincreasing rearrangement of $X$, see Remark \ref{rem-Lorentz}. In that context one defines
\[
\|X\|= \inf\Big\{ a>0 : \int_0^1 \phi\Big(\frac{X^*(u)}{a}\Big)w(u) \d u \leq 1\Big\}.
\]
We consider again the corresponding distortion function $g(u)=\int_0^u w(v)\d v$, $u\in [0,1]$. Then $\Lambda_{\phi,w}=\{X: |X|\in L_g^\phi\}$ and $\|X\|=\pi_{g,\phi,0}(|X|)$ for $X\in \Lambda_{\phi,w}$. However, in general, one cannot recover $L_g^\phi$ from $\Lambda_{\phi,w}$ in the same way as in Remark \ref{rem-Lorentz}, see Example \ref{ex-LorentzLgRhog} in the Appendix.

In the literature, Orlicz-Lorentz spaces are usually studied for decreasing weights $w$. In that case, $\Lambda_{\phi,w}\subset L_g^\phi$, see Proposition \ref{p-OrlLorentzLg} in the Appendix.

(b) Following Remarks \ref{r-percfct} and \ref{r-prime}(a), we interpret $g$ as the decision maker's probability perception function and 
$\phi$ (or rather $-\phi(-\cdot)$) as the decision maker's utility function.
\end{remark}

Since $\phi$ is increasing and continuous, we have that $\phi(F_X^{-1})= F_{\phi(X)}^{-1}$. Thus we see that
\begin{equation}\label{eq-OL}
\pi_{g,\phi,\alpha}(X) = \inf\Big\{a>0 : \rho_g\Big(\phi\Big(\frac{X}{a}\Big)\Big)\leq 1-\alpha\Big\},
\end{equation}
that is, one replaces the expectation by $\rho_g$ in Definition \ref{d-Orlicz}.

By the same argument, using Proposition \ref{p-equiv} via \eqref{eq-defdist}, we obtain the following alternative representations of the integral appearing in the definition of Orlicz-Lorentz premia; note that the value $a$ can be incorporated into $X$.

\begin{proposition}\label{p-equivbis} 
Let $X\geq 0$. 

\emph{(a)} We have
\[
\int_0^1 \phi(F_X^{-1}(1-u)) \d g(u)= \int_0^{\infty} g(\overline{F}_{\phi(X)}(x)-)\d x,
\]
where $\overline{F}_X(x)=1-F_X(x)$.

\emph{(b)} Let $h:[0,1]\to [0,1]$ be given by $h(u)=1-g((1-u)-)$. Then
\[
\int_0^1 \phi(F_X^{-1}(1-u)) \d g(u)= \int_0^\infty x\, \d (h\circ F_{\phi(X)})(x).
\]
\end{proposition}

\begin{remark}\label{r-Choquet} In keeping with Remark \ref{r-prime}(a), let us extend $\phi$ to an increasing convex function on $\mathbb{R}$, define the (concave) utility function $U(t)=-\phi(-t)$ on $\mathbb{R}$, and consider the financial position $Y:=-X$ associated with the risk $X\geq 0$. In decision theory, the Choquet integral
\begin{equation}\label{eq-rdu}
(C)\int U(Y) \d (h\circ P)
\end{equation}
is called the \textit{rank-dependent expected utility} of $Y$ with respect to a distortion function $h$ with $h(0)=0$. This notion was introduced for discrete $Y$ by Quiggin \cite{Qui82}, \cite{Qui93}, see \cite[p.\ 68]{Hei03} for the general formula, and has since been studied extensively in decision theory, see \cite{Wak10}, and more recently also in AI research, see \cite{GoPe20}. 

Now, one has that 
\[
\int_0^1 \phi(F_{X}^{-1}(1-u)) \d g(u)=-(C)\int U(Y) \d (h\circ P),
\]
where $g(u)=1-h((1-u)-)$, see Proposition \ref{p-rdu} in the Appendix. Thus there is a close link between Orlicz-Lorentz premia and rank-dependent expected utility; see also \cite[Section 6.1]{DDGKL06}.
\end{remark}

\subsection{Orlicz-Lorentz: the infimum}
In view of the definition of the Orlicz-Lorentz premia, two questions arise: is the infimum attained, and if so do we have equality in the defining condition at the minimum. In general, the answers are negative.

\begin{example}\label{ex-distOrlicz2}
Clearly, for $X=0$, the infimum is not attained. But this can also happen for nonzero risks. If, for example, $X\in (L_g^\phi)_+$ with  $P(X=0)=\frac{1}{2}$ and $g(u)=\max(2u-1,0)$, then $\int_0^1 \phi\big(\frac{F_X^{-1}(1-u)}{a}\big) \d g(u) = 0$ for all $a>0$, so that $\pi_{g,\phi,\alpha}(X)=0$, and the infimum is not attained.

An example where we do not have equality in the defining condition at the minimum was given in \cite[pp.\ 45-46]{HaGo82}, where $g$ is even the identity function.
\end{example}

In order to obtain positive answers, let $X$ be any positive risk, not necessarily in $(L^\phi_g)_+$, and consider the function $\psi:(0,\infty)\to [0,\infty]$ given by
\[
\psi(x) = \int_0^1 \phi\Big(\frac{F_X^{-1}(1-u)}{x}\Big) \d g(u).
\]
The following lemma generalizes and extends \cite[Lemma 4]{HaGo82}. 

\begin{lemma}\label{l-distOrlicz} 
Let $g$ be a distortion function, $\phi$ a Young function, and $X\geq 0$. Then:
\begin{enumerate}[label=\emph{(\alph*)}] 
\item Either $\{\psi=0\}=\varnothing$ or $\{\psi=0\}= (0,\infty)$. Moreover, $\{\psi=0\}=\varnothing$ if and only if $X\neq 0$ and $g$ is not identically $0$ on $[0,P(X>0))$.
\item If $g=0$ on some neighbourhood of $0$, or if $X\in L^\infty$, then $\{\psi<\infty\}=(0,\infty)$.
\item If $\phi$ satisfies the $\Delta_2$-condition, then either $\{\psi<\infty\}=\varnothing$ or $\{\psi<\infty\}=(0,\infty)$.
\item $\psi$ is right-continuous.
\item $\psi$ is continuous at every interior point of $\{\psi<\infty\}$.
\item $\psi$ is decreasing.
\item $\psi$ is strictly decreasing on $\{0<\psi<\infty\}$.
\item If $\{\psi=0\}=\varnothing$ then $\lim_{x\to 0}\psi(x)=\infty$.
\item If $\{\psi<\infty\}\neq\varnothing$ then $\lim_{x\to\infty}\psi(x)=0$.
\end{enumerate}
\end{lemma}

\begin{proof}
Assertion (d) follows from the monotone convergence theorem, (e) and (i) follow from the dominated convergence theorem, while (f) is obvious. 

(a) If $X=0$ then $\{\psi=0\}= (0,\infty)$. Else suppose that $X\neq 0$, and hence $q:=P(X>0)>0$. Thus $F_X^{-1}(1-u)=0$ for $u\in [q,1)$ and $F_X^{-1}(1-u)>0$ for $u\in [0,q)$. If $g=0$ on $[0,q)$, then $\mu_g([0,q))=0$, where $\mu_g$ is the probability measure induced by $g$. It follows that $\{\psi=0\}= (0,\infty)$. If $g$ is not identically $0$ on $[0,q)$, then $\{\psi=0\}=\varnothing$.

(g) We may assume that $\{0<\psi<\infty\}\neq\varnothing$. By (a), $q:=P(X>0)>0$ and  $g=0$ on $[0,P(X>0))$, so that $\mu_g ([0,q))>0$. Also, as we have seen above, $F_X^{-1}(1-u)=0$ for $u\in [q,1)$, so that $\psi(z) = \int_{[0,q)} \phi(\frac{F_X^{-1}(1-u)}{z}) \d g(u)$ for all $z>0$, and $F_X^{-1}(1-u)>0$ for $u\in [0,q)$. 

Now let $x,y\in \{0<\psi<\infty\}$ and $x<y$. Since $\phi$ is strictly increasing on $\{\phi>0\}$, we have for $u\in [0,q)$ that $\phi\big(\frac{F_X^{-1}(1-u)}{x}\big)>\phi\big(\frac{F_X^{-1}(1-u)}{y}\big)$. Since $\psi(y)<\infty$, this implies that $\psi(x)>\psi(y)$.

(h) Let $q=P(X>0)$. As in (g), the hypothesis implies $q>0$, $\mu_g([0,q))>0$, $\psi(x) = \int_{[0,q)} \phi(\frac{F_X^{-1}(1-u)}{x}) \d g(u)$ for $x>0$, and $F_X^{-1}(1-u)>0$ for $u\in [0,q)$. Then the claim follows from the monotone convergence theorem.
 
(b) Suppose that $g=0$ on $[0,u_0)$ for some $u_0\in (0,1)$. Then $\psi(x)= \int_{[u_0,1]}\phi(\frac{F_X^{-1}(1-u)}{x}) \d g(u)<\infty$ for all $x>0$. If $X\in L^\infty$, then $F_X^{-1}$ is bounded on $(0,1]$ and therefore $\{\psi<\infty\}=(0,\infty)$. 

(c) Suppose that there is some $a>0$ such that $\psi(a)<\infty$. Then it follows from the $\Delta_2$-condition that, for some $s\geq 0$ and $K>0$,  $\phi(\frac{y}{a/2^n})\leq K^n \phi(\frac{y}{a})$ for all $y\geq as$ and hence $\psi(\frac{a}{2^n})<\infty$, for all $n\geq 1$. Thus, (f) implies that $\{\psi<\infty\}=(0,\infty)$.
\end{proof}

Part (a) of the following lemma gives a partial converse of property (c) above, part (b) is for later use. The proof is inspired by that of \cite[Theorem 133.4]{Zaa83}.

\begin{lemma}\label{l-conv}
Suppose that the underlying probability space $(\Omega,\mathcal{A},P)$ is atomless. Let $g$ be a distortion function with $g(0)=0$ and $g>0$ on $(0,1]$ that is continuous on some neighbourhood of $0$, and let $\phi$ be a Young function that does not satisfy the $\Delta_2$-condition.

\emph{(a)} Then there is a risk $X\geq 0$ on $\Omega$ and $y>x>0$ such that $\psi(x)=\infty$ and $\psi(y)<\infty$.

\emph{(b)} There are risks $X_n\in (L_g^\phi)_+$ such that $X_n\searrow 0$, but $\pi_{g,\phi,\alpha}(X_n)\geq \frac{1}{2}$ for all $n$.
\end{lemma}

\begin{proof}
(a) First, if $\phi$ does not satisfy the $\Delta_2$-condition, there is a strictly increasing positive sequence $(t_n)_n$ such that $\phi(2t_n)\geq n \phi(t_n)$ and $\phi(t_n)\geq 1$, $n\geq 1$.

Now, by assumption, there is some $u_0\in (0,1]$ such that $g(u_0)>0$ and $g:[0,u_0]\to [0,g(u_0)]$ is continuous and hence surjective. 

Next choose a strictly positive sequence $(a_n)_n$ such that $\sum_{n=1}^\infty a_n = g(u_0)$ and $\sum_{n=1}^\infty na_n=\infty$. By surjectivity, there is a strictly decreasing sequence $(b_n)_{n\geq 0}$ in $(0,u_0]$ such that $g(b_n)=\sum_{k=n+1}^\infty \frac{a_k}{\phi(t_k)}$, $n\geq 0$. Since $g>0$ on $(0,1]$, we have that $b_n \to 0$.

Finally, since $P$ is atomless, there exists a pairwise disjoint sequence $(A_n)_{n\geq 1}$ of sets in $\mathcal{A}$ with $P(A_n)=b_{n-1}-b_n$, $n\geq 1$; see \cite[Theorem 8.14.2]{GeGe25}. Consider the risk $X=\sum_{n=1}^\infty t_n \mathds{1}_{A_n}$. Then $F_X(x)=1-b_{n-1}$ for $t_{n-1}\leq x < t_n$, $n\geq 1$, where $t_0=0$. Thus we have that
\[
\psi(1) = \sum_{n=1}^\infty \phi(t_n) (g(b_{n-1})-g(b_n))=\sum_{n=1}^\infty a_n <\infty,
\]
where we have used that $g$ is continuous at each $b_n$; in the same way, 
\[
\psi(\tfrac{1}{2}) = \sum_{n=1}^\infty \phi(2t_n) \frac{a_n}{\phi(t_n)}\geq \sum_{n=1}^\infty na_n =\infty.
\]
This proves the claim.

(b) Consider the risk $X=\sum_{n=1}^\infty t_n \mathds{1}_{A_n}$ of part (a), and let $X_n=\sum_{k=n}^\infty t_k \mathds{1}_{A_k}$, so that $X_n\searrow 0$. Then $X=X_1$ satisfies $\psi(1)<\infty$, so that $X_n\in (L_g^\phi)_+$ for all $n$. Also,  
\begin{align*}
\int_0^1 \phi\Big(\frac{F_{X_n}^{-1}(1-u)}{1/2}\Big) \d g(u) = \sum_{k=n}^\infty \phi(2t_k) \frac{a_k}{\phi(t_k)}\geq \sum_{k=n}^\infty ka_k =\infty,
\end{align*}
so that $\pi_{g,\phi,\alpha}(X_n)\geq \tfrac{1}{2}$ for all $n$, which had to be shown.
\end{proof}

Assertion (a) of the following result now characterizes when the infimum in the definition of the Orlicz-Lorentz premium is attained. 

\begin{proposition}\label{p-distOrlicz}
Let $X\in (L_g^\phi)_+$. 

\emph{(a)} We have that $\pi_{g,\phi,\alpha}(X)\neq 0$ if and only if $X\neq 0$ and $g$ is not identically $0$ on $[0,P(X>0))$. In that case,
\begin{align*}
\pi_{g,\phi,\alpha}(X)= \min\Big\{a>0 : \int_0^1 \phi\Big(\frac{F_X^{-1}(1-u)}{a}\Big) \d g(u)\leq 1-\alpha\Big\}.
\end{align*} 

\emph{(b)} If $a>0$ satisfies
\[
\int_0^1 \phi\Big(\frac{F_X^{-1}(1-u)}{a}\Big) \d g(u)= 1-\alpha,
\]
then $a=\pi_{g,\phi,\alpha}(X)$.

\emph{(c)} Suppose that $\pi_{g,\phi,\alpha}(X)\neq 0$. If $g=0$ on some neighbourhood of $0$, or if $X\in L^\infty$, or if $\phi$ satisfies the $\Delta_2$-condition, then there is a unique value $a>0$ such that 
\[
\int_0^1 \phi\Big(\frac{F_X^{-1}(1-u)}{a}\Big) \d g(u)= 1-\alpha,
\]
and $a=\pi_{g,\phi,\alpha}(X)$.
\end{proposition}

\begin{proof}
(a) If $X=0$, or if $X\neq 0$ and $g=0$ on $[0,P(X>0))$, then $\psi(x)=0$ for all $x>0$ by Lemma \ref{l-distOrlicz}(a), so that $\pi_{g,\phi,\alpha}(X)=0$. Otherwise, the result follows from the points (a), (d), (h) and (i) of Lemma \ref{l-distOrlicz}; note that $\{\psi<\infty\}\neq\varnothing$ because $X\in (L_g^\phi)_+$.

(b) follows from Lemma \ref{l-distOrlicz}(g).

(c) Again, $\{\psi<\infty\}\neq\varnothing$ because $X\in (L_g^\phi)_+$. Thus, Lemma \ref{l-distOrlicz}(b) and (c) imply that $\{\psi<\infty\}=(0,\infty)$. Then existence follows from points (a), (e), (h) and (i) of Lemma \ref{l-distOrlicz}. And uniqueness follows from (b) above. 
\end{proof}

In other words, under the assumptions stated in (c), one can define $\pi_{g,\phi,\alpha}(X)$ as the unique value satisfying
\begin{equation}\label{eq-unique}
\int_0^1 \phi\Big(\frac{F_X^{-1}(1-u)}{\pi_{g,\phi,\alpha}(X)}\Big) \d g(u)= 1-\alpha.
\end{equation}
This is the case, in particular, if $g$ is the identity function, $X\neq 0$, and either $X$ is bounded or $\phi$ satisfies the $\Delta_2$-condition, so that we recover the findings in \cite[Remark 3]{HaGo82} and \cite[p.\ 108]{BeRo12}.

\begin{remark} In analogy to the so-called Orlicz hearts, see \cite{BeRo12}, \cite{EdSo92}, \cite[Section 3.4, Definition 2]{RaRe91}, one might define the \textit{heart} $M_g^\phi$ of $L_g^\phi$ as the space of all risks $X$ for which 
\[
\int_0^1 \phi\Big(\frac{|F_X^{-1}(1-u)|}{a}\Big) \d g(u) <\infty
\]
holds \textit{for all} $a>0$. It follows as in the proof of Lemma \ref{l-distOrlicz}(b) that
\[
L^\infty \subset M_g^\phi \subset L_g^\phi. 
\]
Moreover, by the proof of Lemma \ref{l-distOrlicz}(c), $M_g^\phi = L_g^\phi$ if $\phi$ satisfies the $\Delta_2$-condition. Now, several results in this paper that depend on the $\Delta_2$-condition do in fact hold in $M_g^\phi$ for any $\phi$. For example, identity \eqref{eq-unique} holds for all $X\in M_g^\phi$ provided that $\pi_{g,\phi,\alpha}(X)\neq 0$.

Since we are mainly interested in results that hold on all of $L_g^\phi$, we do not pursue this aspect here.
\end{remark}

\subsection{Orlicz-Lorentz: risk theoretic properties}
We obtain several properties of general Orlicz-Lorentz premia.

\begin{proposition}\label{p-distOrlicz2}~

\begin{enumerate}[label=\emph{(\alph*)}]
\item For any $X\in (L_g^\phi)_+$,
\[
\frac{\rho_g(X)}{\phi^{-1}(1-\alpha)}\leq \pi_{g,\phi,\alpha}(X)\leq \frac{\esssup X}{\phi^{-1}(1-\alpha)}.
\]
\item For any $b\geq 0$, $\pi_{g,\phi,\alpha}(b)= \frac{b}{\phi^{-1}(1-\alpha)}$.
\end{enumerate}
\end{proposition}

\begin{proof}
(a) First note that, by the assumptions on $\phi$, $\phi^{-1}(1-\alpha)>0$ is well defined and $\phi(\phi^{-1}(1-\alpha))=1-\alpha$. 

The first inequality is trivial if $\rho_g(X)=0$. Else let $0<\eps<\rho_g(X)$. Then Jensen's inequality implies by convexity of $\phi$ that
\begin{align*}
\int_0^1 \phi\Big(\frac{F_X^{-1}(1-u)}{(\rho_g(X)-\eps)/\phi^{-1}(1-\alpha)}\Big)\d g(u)&\geq \phi\Big(\frac{\phi^{-1}(1-\alpha)}{\rho_g(X)-\eps}\int_0^1 F_X^{-1}(1-u)\d g(u)\Big)\\
&=\phi\Big(\frac{\phi^{-1}(1-\alpha)}{\rho_g(X)-\eps}\rho_g(X)\Big)>1-\alpha,
\end{align*}
where we have used that $\phi$ is strictly increasing on $\{\phi>0\}$. Thus $\pi_{g,\phi,\alpha}(X)\geq \frac{\rho_g(X)-\eps}{\phi^{-1}(1-\alpha)}$ for any $\eps>0$, which implies the first inequality.

The second inequality is trivial if $\esssup X=\infty$. Otherwise we use the fact that $F_X^{-1}$ is bounded by $\esssup X$ and take $a= \frac{\esssup X}{\phi^{-1}(1-\alpha)}$.

(b) follows directly from the fact that $F^{-1}_b=b$ on $(0,1]$ and Proposition \ref{p-distOrlicz}(b).
\end{proof}

\begin{proposition}\label{p-distOrliczbis}
The Orlicz-Lorentz premium $\pi_{g,\phi,\alpha}$ is monotonic and positively homogeneous on $(L_g^\phi)_+$.
\end{proposition}

\begin{proof}
The monotonicity follows from the monotonicity of $\phi$ and $F_X^{-1}$. The positive homogeneity follows from the fact that, for $\lambda >0$, $\phi\big(\frac{F_{\lambda X}^{-1}(1-u)}{a}\big)=\phi\big(\frac{F_X^{-1}(1-u)}{a/\lambda}\big)$; note also that $\pi_{g,\phi,\alpha}(0)=0$.
\end{proof} 

\begin{remark}
It follows that the Orlicz-Lorentz premium $\pi_{g,\phi,\alpha}$ is a return risk measure (see Remark \ref{r-return}) if $\phi$ is normalized and $\alpha=0$. Just like for Orlicz premia, see \cite{ABL25}, one obtains return risk measures even for many non-convex functions $\phi$. For example, if $\phi(t)=t^p$, $0<p<1$, and $g(u)=\min(\frac{u}{1-\beta},1)$, $0<\beta<1$, then $\pi_{g,\phi,0}(X)=(\frac{1}{1-\beta}\int_\beta^1 F_X^{-1}(u)^p\d u)^{1/p}$ for $X\in L_+^\infty$, which is a return risk measure. We will not pursue this here. Interestingly, Orlicz-Lorentz spaces have also been studied for non-convex functions $\phi$, see \cite[Section 5]{KMP03}.
\end{remark}

We will next show that Orlicz-Lorentz premia are subadditive for comonotonic risks; unlike for the distortion risk measures, see Proposition
\ref{p-distcomon}, one cannot expect additivity here because Orlicz premia already fail to have this property. For a concrete counter-example, take $\phi(t)=t^2$, any $\alpha<1$, $X=\mathds{1}_{[0,\frac{1}{2})}(U)$ and $Y=\mathds{1}_{[\frac{1}{2},1]}(U)$, where $U$ is uniformly distributed on $[0,1]$.

\begin{proposition}\label{p-distOrliczcom}
Let $X,Y\in (L_g^\phi)_+$ be comonotonic risks. Then 
$X+Y\in (L_g^\phi)_+$ and
\[
\pi_{g,\phi,\alpha}(X+Y)\leq \pi_{g,\phi,\alpha}(X)+\pi_{g,\phi,\alpha}(Y).
\]
\end{proposition}

\begin{proof}
Let $\varepsilon>0$. Then there are $a_1,a_2>0$ with $a_1<\pi_{g,\phi,\alpha}(X)+\varepsilon$ and $a_2<\pi_{g,\phi,\alpha}(Y)+\varepsilon$
such that $\int_0^1 \phi(\frac{F_{X}^{-1}(1-u)}{a_1})\d g(u)\leq 1-\alpha$ and $\int_0^1 \phi(\frac{F_{Y}^{-1}(1-u)}{a_2})\d g(u)\leq 1-\alpha$. 

Now, by comonotonic additivity of $\VaR$, we have that $F_{X+Y}^{-1} = F_{X}^{-1}+F_{Y}^{-1}$ and therefore, using the convexity of $\phi$,
\begin{align*}
\phi\Big(\frac{F_{X+Y}^{-1}(1-u)}{a_1+a_2}\Big)&=\phi\Big(\frac{a_1}{a_1+a_2}\frac{F_{X}^{-1}(1-u)}{a_1}+\frac{a_2}{a_1+a_2}\frac{F_{Y}^{-1}(1-u)}{a_2}\Big)\\
&\leq \frac{a_1}{a_1+a_2}\phi\Big(\frac{F_{X}^{-1}(1-u)}{a_1}\Big)+\frac{a_2}{a_1+a_2}\phi\Big(\frac{F_{Y}^{-1}(1-u)}{a_2}\Big).
\end{align*}
Integrating with respect to $\d g$ we obtain by the properties of $a_1$ and $a_2$ that
\[
\int_0^1 \phi\Big(\frac{F_{X+Y}^{-1}(1-u)}{a_1+a_2}\Big)\d g(u)\leq 1-\alpha,
\]
which implies that $X+Y\in (L_g^\phi)_+$ and
\[
\pi_{g,\phi,\alpha}(X+Y)\leq a_1+a_2.
\]
Since $\varepsilon>0$ is arbitrary, the result follows.
\end{proof}

As for the Fatou properties, we have the following results.

\begin{proposition}\label{p-distOrliczFatou}
The Orlicz-Lorentz premium $\pi_{g,\phi,\alpha}$ has the Fatou property on $(L_g^\phi)_+$.
\end{proposition}

\begin{proof}
By Remark \ref{r-Fatou}(c) and Proposition \ref{p-distOrliczbis} it suffices to show that if $X_n\nearrow X$ and $X_1,X\in (L_g^\phi)_+$ then $\pi_{g,\phi,\alpha}(X_n)\to\pi_{g,\phi,\alpha}(X)$, or, equivalently, $\pi_{g,\phi,\alpha}(X)\leq\sup_n\pi_{g,\phi,\alpha}(X_n)$.

We first note that by \eqref{eq-defdist} and a property of quantile functions we have for $X\in (L_g^\phi)_+$ and $a>0$ 
\[
\int_0^1 \phi\Big(\frac{F_X^{-1}(1-u)}{a}\Big) \d g(u) = \int_0^{\infty} g\big(\overline{F}_{\phi(\frac{X}{a})}(x)-\big)\d x.
\]

Let $X_n\nearrow X$ with $X_1,X\in (L_g^\phi)_+$. As in the proof of Proposition \ref{p-distFatou} one deduces that 
\begin{equation}\label{eq-Fat}
\int_0^1 \phi\Big(\frac{F_{X_n}^{-1}(1-u)}{a}\Big) \d g(u)\nearrow\int_0^1 \phi\Big(\frac{F_X^{-1}(1-u)}{a}\Big) \d g(u).
\end{equation}

Take $a= \sup_n\pi_{g,\phi,\alpha}(X_n)$ and $\varepsilon>0$. By definition, we have for any $n$,
\[
\int_0^1 \phi\Big(\frac{F_{X_n}^{-1}(1-u)}{a+\varepsilon}\Big) \d g(u)\leq 1-\alpha.
\]
By \eqref{eq-Fat} we find that $\int_0^1 \phi(\frac{F_X^{-1}(1-u)}{a+\varepsilon}) \d g(u)\leq 1-\alpha$ and thus $\pi_{g,\phi,\alpha}(X)\leq a+\varepsilon$. Since $\varepsilon>0$ is arbitrary, the claim follows.
\end{proof}

In the sequel, the following property ($P_{g,\phi}$) will be crucial:\\[.5em]
\begin{minipage}{.8\textwidth}%
      -- $g(0)=0$,\\  -- $g$ is continuous, and\\ -- \textit{either} $g=0$ on some neighbourhood of $0$\\ \phantom{--} \textit{or} $\phi$ satisfies the $\Delta_2$-condition.
\end{minipage}

\begin{proposition}\label{p-OLrevFatou}~ 

\emph{(a)} If \emph{($P_{g,\phi}$)} holds, then $\pi_{g,\phi,\alpha}$ has the reverse Fatou property on $(L_g^\phi)_+$, and hence the Lebesgue property.

\emph{(b)} If the underlying probability space $(\Omega,\mathcal{A},P)$ is atomless, then $\pi_{g,\phi,\alpha}$ has the reverse Fatou property on $(L_g^\phi)_+$ if and only if \emph{($P_{g,\phi}$)} holds.
\end{proposition}

\begin{proof}
(a) It suffices, by Remark \ref{r-Fatou}(c) and Proposition \ref{p-distOrliczbis}, to prove that if $X_n\searrow X$ and $X_1, X \in (L_g^\phi)_+$ then $\pi_{g,\phi,\alpha}(X)\geq\inf_n\pi_{g,\phi,\alpha}(X_n)$. This is trivial if $a:=\inf_n\pi_{g,\phi,\alpha}(X_n)=0$. So suppose that $a>0$, and let $a_n= \pi_{g,\phi,\alpha}(X_n)$, $n\geq 1$. Since $a_n>0$, Proposition \ref{p-distOrlicz}(c) with ($P_{g,\phi}$) implies that
\[
\int_0^1 \phi\Big(\frac{F_{X_n}^{-1}(1-u)}{a_n}\Big) \d g(u) = 1-\alpha
\]
for all $n$. Hence, by \eqref{eq-defdist}, a property of quantile functions, the continuity of $g$, and the fact that $g(0-)=g(0)=0$,
\[
\int_0^{\infty} g\big(\overline{F}_{\phi(\frac{X_n}{a_n})}(x)\big)\d x=1-\alpha
\]
for all $n$. Now since $X_n\to X$, $a_n\to a \neq 0$, and $\phi$ is continuous, we have that $\phi(\frac{X_n}{a_n})\to \phi(\frac{X}{a})$, and hence $\overline{F}_{\phi(\frac{X_n}{a_n})}(x)\to \overline{F}_{\phi(\frac{X}{a})}(x)$ for all $x\geq 0$ with at most countably many exceptions. In view of the continuity of $g$ we deduce that 
\[
\int_0^{\infty} g\big(\overline{F}_{\phi(\frac{X}{a})}(x)\big)\d x=1-\alpha,
\]
where we have used the dominated convergence theorem; note that $0\leq X_n\leq X_1$ and, by Lemma \ref{l-distOrlicz}(b) and (c), $\int_0^{\infty} g\big(\overline{F}_{\phi(\frac{X_1}{a})}(x)\big)\d x<\infty$.

Now, by Proposition \ref{p-distOrlicz}(b), we obtain that $\pi_{g,\phi,\alpha}(X)=a$.

(b) Suppose that $\pi_{g,\phi,\alpha}$ possesses the reverse Fatou property. 

First, suppose that $g$ is not continuous or that $g(0)\neq 0$. Then there is some $u\in\mathopen[0,1)$ such that $g(u-)<g(u)$; recall that $g(0-)=0$. Let $(p_n)_n$ be a strictly decreasing sequence in $[0,1]$ with limit $u$. As in the proof of Proposition \ref{p-distrevFatou}, we define $X_n=\mathds{1}_{A_n}$ and $X=\mathds{1}_{A}$, where $(A_n)_n$ is a decreasing sequence of sets in $\mathcal{A}$ with $P(A_n)=p_n$, $n\geq 1$, and $A:=\bigcap_{n=1}^\infty A_n$, which satisfies $P(A)=u$. Then the $X_n$ belong to $(L_g^\phi)_+$ as bounded risks, and $X_n\to X$ on $\Omega$. Also, $\pi_{g,\phi,\alpha}(X_n)=\inf\{a>0: \phi(\frac{1}{a})g(p_n-)\leq 1-\alpha\}$, hence
\[
\pi_{g,\phi,\alpha}(X_n)=\frac{1}{\phi^{-1}\Big(\frac{1-\alpha}{g(p_n-)}\Big)}\geq \frac{1}{\phi^{-1}\Big(\frac{1-\alpha}{g(u)}\Big)}
\]
for all $n$, while $\pi_{g,\phi,\alpha}(X)= 1/\phi^{-1}(\frac{1-\alpha}{g(u-)})$ for $u>0$ and $\pi_{g,\phi,\alpha}(X)= 0$ for $u=0$. Since $\phi^{-1}$ is strictly increasing on $(0,\infty)$, we see that $\pi_{g,\phi,\alpha}(X)<\limsup_{n\to\infty}\pi_{g,\phi,\alpha}(X_n)$, contradicting the reverse Fatou property. So we have that $g$ is continuous and $g(0)=0$.

Secondly, suppose that $g>0$ on $(0,1]$ and that $\phi$ does not satisfy the $\Delta_2$-condition. Then, by Lemma \ref{l-conv}(b), there are risks $X_n\in (L_g^\phi)_+$ such that $X_n\searrow 0$ and $\pi_{g,\phi,\alpha}(X_n)\geq \frac{1}{2}$ for all $n$. This contradicts the reverse Fatou property.
\end{proof}

When we decide to work on $L^\infty_+$, Proposition \ref{p-distOrlicz}(c) tells us that we do not need to demand that $g$ is constant on some neighbourhood of $0$ or that $\phi$ satisfies the $\Delta_2$-condition. Thus the same proof as above yields the following.

\begin{proposition}\label{p-OLrevFatoubdd}~

\emph{(a)} If $g(0)=0$ and $g$ is continuous, then $\pi_{g,\phi,\alpha}$ has the reverse Fatou property on $L^\infty_+$, and hence the Lebesgue property.

\emph{(b)} If the underlying probability space $(\Omega,\mathcal{A},P)$ is atomless, then $\pi_{g,\phi,\alpha}$ has the reverse Fatou property on $L^\infty_+$ if and only if $g(0)=0$ and $g$ is continuous.
\end{proposition}

Also, $\pi_{g,\phi,\alpha}$ quite trivially preserves dominance in the stochastic order. Recall that a risk $X$ is said to be \textit{smaller} than a risk $Y$ \textit{in stochastic dominance} (or \textit{in stochastic order}), denoted as $X\leq_{\text{st}} Y$, if
\[
F_X(x)\geq F_Y(x)\quad\text{for all $x\in \R$},
\]
which is equivalent to saying that $F_X^{-1}(u)\leq F_Y^{-1}(u)$ for all $u\in (0,1)$; see \cite{DDGK05}, \cite{DVGKTV06}, \cite{ShSh07}.

Thus we have:

\begin{proposition}\label{p-st}
Let $Y\in (L_g^\phi)_+$ and $X\geq 0$. Then 
\[
X\leq_{\emph{\text{st}}} Y \Longrightarrow X\in (L_g^\phi)_+ \text{ and } \pi_{g,\phi,\alpha}(X)\leq \pi_{g,\phi,\alpha}(Y).
\]
\end{proposition}

\subsection{Orlicz-Lorentz: the concave case}
In order to deduce stronger properties of the Orlicz-Lorentz premia we will now demand that $g$ be concave.

We first need the following technical result on general positive risks, which is an immediate consequence of a well known representation of the Tail Value at Risk. 

\begin{lemma}\label{l-TVaR}
Let $\phi$ be a Young function, $X\geq 0$ a risk, $a>0$ and $0<\beta\leq 1$. Then, for any $x\geq 0$,
\begin{align*}
\int_0^\beta \phi\Big(\frac{F_X^{-1}(1-u)}{a}\Big) \d u\leq  E\Big(\Big(\phi\Big(\frac{X}{a}\Big)-x\Big)^+\Big)+\beta x,
\end{align*}
with equality at $x=\phi\big(\frac{F_X^{-1}(1-\beta)}{a}\big)$ if $\beta<1$ and $x=0$ if $\beta=1$.
\end{lemma}

\begin{proof}
Let $X$ be a positive risk. If $X\in L^1$ and $0<\alpha<1$, we have by Proposition \ref{p-TVaR} that, for any $x\geq 0$,
\[
\frac{1}{1-\alpha}\int_\alpha^1 F_X^{-1}(u)\d u \leq \frac{1}{1-\alpha}E\big((X-x)^+\big)+x,
\]
with equality at $x=F_X^{-1}(\alpha)$. If $\alpha=0$ then
\[
\int_\alpha^1 F_X^{-1}(u)\d u =E(X) \leq E\big((X-x)^+\big)+x
\]
for any $x\geq 0$ since $y\leq (y-x)^++x$, $y\in \R$, and we have equality for $x=0$.  

If $E(X)=\infty$, both sides of these inequalities are infinite for any $x\geq 0$. 

Writing $\beta=1-\alpha$ and replacing $X$ by $\phi\big(\frac{X}{a}\big)$ then proves the claim, where we note that $F^{-1}_{\phi(X/a)} = \phi(F_X^{-1}/a)$.
\end{proof}

We can now show that $\pi_{g,\phi,\alpha}$ preserves stop-loss order. Here, a risk $X$ is said to be \textit{smaller} than a risk $Y$ \textit{in stop-loss order} (or \textit{in increasing convex order}), denoted as $X\leq_{\text{sl}} Y$, if 
\[
E((X-d)^+)\leq E((Y-d)^+)\quad \text{for all $d\in \R$},
\]
see \cite{DDGK05}, \cite{DVGKTV06}, \cite{ShSh07}. If $X$ and $Y$ are positive, this is equivalent to saying that $E(\varphi(X))\leq E(\varphi(Y))$ for all increasing convex functions $\varphi$ on $\mathbb{R}$ for which the expectations exist, see \cite[Theorem 4.A.2]{ShSh07}. 

\begin{proposition}\label{p-distOrliczsl}
Let $g$ be concave. Let $X\geq 0$ and $Y\in (L_g^\phi)_+$. Then 
\[
X\leq_{\emph{\text{sl}}} Y \Longrightarrow X\in (L_g^\phi)_+ \text{ and } \pi_{g,\phi,\alpha}(X)\leq \pi_{g,\phi,\alpha}(Y).
\]
\end{proposition}

\begin{proof}
Let $X\geq 0$ and $Y\in (L_g^\phi)_+$ with $X\leq_{\text{sl}} Y$. 

Let us first assume that $g$ is concave and piecewise linear. Then there are $0<\beta_1<\beta_2<\ldots<\beta_n=1$, $c_0\geq 1$ and $c_k>0$, $k=1,\ldots,n,$ such that
\[
g(u)=c_0+\sum_{k=1}^n c_k \min\Big(\frac{u}{\beta_k},1\Big), u\in [0,1].
\]
Since $c_0+\sum_{k=1}^n c_k=1$, $g$ is a convex combination of the functions $g_0= 1$ and $g_k(u)=\min\big(\frac{u}{\beta_k},1\big)$, $k=1,\ldots,n$. Thus, for any $a>0$,
\begin{align*}
\int_0^1  \phi\Big(\frac{F_X^{-1}(1-u)}{a}\Big) \d g(u)&= \sum_{k=0}^n c_k\int_0^1  \phi\Big(\frac{F_X^{-1}(1-u)}{a}\Big) \d g_k(u)\\
&= c_0\phi \Big(\frac{F_X^{-1}(1)}{a}\Big) + \sum_{k=1}^n \frac{c_k}{\beta_k}\int_0^{\beta_k}  \phi\Big(\frac{F_X^{-1}(1-u)}{a}\Big) \d u,
\end{align*}
and similarly for $Y$. Note that if $c_0>0$ then $X$ is bounded above because $X\in (L_g^\phi)_+$, see the discussion after Definition \ref{d-OL}. Now, $X\leq_{\text{sl}} Y$ implies, in particular, that $\esssup X\leq \esssup Y$ (use the argument in \cite[Section 3.4.1.1]{DDGK05}), hence $F_X^{-1}(1)\leq F_Y^{-1}(1)<\infty$. Thus, the first term is defined. On the other hand, if $c_0=0$, we take it to be zero.  

Let $x_k= \phi\big(\frac{F_Y^{-1}(1-\beta_k)}{a}\big)$, $k=1,\ldots,n-1$, and $x_n=0$. It then follows with Lemma \ref{l-TVaR} that
\begin{align*}
\int_0^1 \phi\Big(\frac{F_X^{-1}(1-u)}{a}\Big) \d g(u) \leq c_0\phi \Big(\frac{F_X^{-1}(1)}{a}\Big) + \sum_{k=1}^n \frac{c_k}{\beta_k}\Big(E\Big(\Big(\phi\Big(\frac{X}{a}\Big)-x_k\Big)^+\Big)+\beta_k x_k\Big).
\end{align*}

As we have seen, $F_X^{-1}(1)\leq F_Y^{-1}(1)$. Also, since $t\mapsto \big(\phi\big(\frac{t}{a}\big)-x\big)^+$ is an increasing convex function on $[0,\infty)$, the stop-loss order implies that $E\big(\big(\phi\big(\frac{X}{a}\big)-x_k\big)^+\big)\leq E\big(\big(\phi\big(\frac{Y}{a}\big)-x_k\big)^+\big)$ for $k=1,\ldots,n$. Thus,
\begin{align*}
\int_0^1 \phi\Big(\frac{F_X^{-1}(1-u)}{a}\Big) \d g(u) &\leq c_0\phi \Big(\frac{F_Y^{-1}(1)}{a}\Big) + \sum_{k=1}^n \frac{c_k}{\beta_k}\Big(E\Big(\Big(\phi\Big(\frac{Y}{a}\Big)-x_k\Big)^+\Big)+\beta_k x_k\Big)\\
&=\int_0^1 \phi\Big(\frac{F_Y^{-1}(1-u)}{a}\Big) \d g(u),
\end{align*}
where the last equality follows from Lemma \ref{l-TVaR} with the definition of the $x_k$.

To finish the proof, let $g$ be an arbitrary concave distortion function. Then there exists an increasing sequence $(g_n)_n$ of piecewise linear concave distortion functions that tends pointwise to $g$ as $n\to\infty$, and hence also $g_n(u-)\to g(u-)$ for all $u\in[0,1]$. 
If $X\in (L_g^\phi)_+$, then also $X\in (L_{g_n}^\phi)_+$ for all $n$. 

Using \eqref{eq-defdist}, a property of quantile functions, and the monotone convergence theorem, we then get, for any $a>0$,
\begin{align*}
\int_0^1 \phi\Big(\frac{F_X^{-1}(1-u)}{a}\Big) \d g_n(u)&=\int_0^\infty  g_n\big(\overline{F}_{\phi(\frac{X}{a})}(x)-\big)\d x\\
&\to \int_0^\infty  g\big(\overline{F}_{\phi(\frac{X}{a})}(x)-\big)\d x= \int_0^1 \phi\Big(\frac{F_X^{-1}(1-u)}{a}\Big) \d g(u).
\end{align*}
Since the same holds for $Y$, the previous inequality for piecewise linear concave distortion functions implies that
\begin{align*}
\int_0^1 \phi\Big(\frac{F_X^{-1}(1-u)}{a}\Big) \d g(u)\leq \int_0^1 \phi\Big(\frac{F_Y^{-1}(1-u)}{a}\Big) \d g(u).
\end{align*}
Since this holds for all $a>0$, we finally deduce that
\[
\pi_{g,\phi,\alpha}(X)\leq \pi_{g,\phi,\alpha}(Y).\qedhere
\]
\end{proof}

This now leads to a simple proof that for concave distortion functions, the Orlicz-Lorentz premia are subadditive. Indeed, it was shown in \cite[Corollary 8]{WaDh98} that any risk measure that preserves stop-loss and is additive for comonotonic risks is subadditive for arbitrary risks. But the proof for this given in \cite[Theorem 4.2.2]{DVGKTV06} also works if the risk measure is only subadditive for comonotonic risks. For the sake of completeness, we give the proof here; recall that $X =_d Y$ means that $X$ and $Y$ have the same distribution.

\begin{lemma}\label{l-risksub}
Suppose that the underlying probability space $(\Omega,\mathcal{A},P)$ is atomless. Let $\mathcal{X}$ be a set of positive risks on $\Omega$ that contains the constants and $\rho:\mathcal{X}\to\mathbb{R}$ a risk measure such that
\begin{enumerate}[label=\emph{(\roman*)}] 
\item $X\geq 0$, $Y\in \mathcal{X}$, $X =_d Y$ $\Longrightarrow$ $X\in\mathcal{X}$ and $\rho(X)= \rho(Y)$;
\item $X,Y\in\mathcal{X}$ comonotonic $\Longrightarrow$ $X+Y\in\mathcal{X}$ and $\rho(X+Y)\leq \rho(X)+\rho(Y)$;
\item $X\geq 0$, $Y\in\mathcal{X}$, $X\leq_{\emph{\text{sl}}} Y$ $\Longrightarrow X\in \mathcal{X}$ and $\rho(X)\leq \rho(Y)$.
\end{enumerate}
Then, for all $X,Y\in\mathcal{X}$, $X+Y\in\mathcal{X}$ and $\rho(X+Y)\leq \rho(X)+\rho(Y)$.
\end{lemma}

\begin{proof}
Let $X,Y\in\mathcal{X}$. Since $P$ has no atoms, there is a random variable $U$ on $\Omega$ that is uniformly distributed on $(0,1)$, see \cite[Proposition A.31]{FoSc16}. Then $X^c:=F_X^{-1}(U)$ and $X$ have the same distribution, as do $Y^c:=F_Y^{-1}(U)$ and $Y$, see \cite[Property 1.5.20]{DDGK05}. By (i), $X^c, Y^c\in\mathcal{X}$. Now, $X^c$ and $Y^c$ are comonotonic, so that, by (ii) with (i), $X^c+ Y^c\in\mathcal{X}$ and $\rho(X^c+Y^c)\leq \rho(X)+\rho(Y)$. Moreover, we have that $X+Y \leq_{\text{sl}} X^c+Y^c$, see \cite[Theorem 7]{DDGKV02} or \cite[Proposition 1]{KDG000}. Thus, by (iii), $X+Y\in\mathcal{X}$ and $\rho(X+Y)\leq \rho(X^c+Y^c) \leq \rho(X)+\rho(Y)$, as had to be shown.
\end{proof}

Thus we obtain the main result of this section.

\begin{theorem}\label{t-distOrlicz}
If $g$ is concave, then $(L_g^\phi)_+$ is a convex cone, and the Orlicz-Lorentz premium $\pi_{g,\phi,\alpha}$ is subadditive on $(L_g^\phi)_+$.
\end{theorem}

\begin{proof} 
We first assume that the underlying probability space $(\Omega,\mathcal{A},P)$ is atomless. Then, by Propositions \ref{p-distOrliczcom} and \ref{p-distOrliczsl}, $\pi_{g,\phi,\alpha}$ satisfies assumptions (ii) and (iii) of Lemma \ref{l-risksub}, while assumption (i) obviously holds. Thus $(L_g^\phi)_+$ is invariant under taking sums, and $\pi_{g,\phi,\alpha}$ is subadditive. Since $(L_g^\phi)_+$ is also clearly invariant under positive scalar multiplication, it is a convex cone.

There is a slight technical problem if $P$ is not atomless. However, by \cite[Example 8.14.3]{GeGe25}, the product space given by $\widetilde{\Omega}=\Omega\times[0,1]$, $\widetilde{\mathcal{A}}=\mathcal{A}\otimes\mathcal{B}[0,1]$, $\widetilde{P}=P\otimes m$, is atomless, where $m$ is the Lebesgue measure. Then the mapping $(L_g^\phi)_+(\Omega)\to (L_g^\phi)_+(\widetilde{\Omega})$, $X\mapsto \widetilde{X}$ with $\widetilde{X}(\omega,u)=X(\omega)$ for $(\omega,u)\in \Omega\times[0,1]$, allows to transfer the result from $(L_g^\phi)_+(\widetilde{\Omega})$ to $(L_g^\phi)_+(\Omega)$; note that $\widetilde{X+Y}=\widetilde{X}+\widetilde{Y}.$
\end{proof}

We have followed here the strategy of proof from \cite[Section 5.2]{DVGKTV06} or \cite[Corollary 8]{WaDh98}; a different, self-contained proof of Theorem \ref{t-distOrlicz} was given by the first author in \cite{Gou22}.

\subsection{Distortion Haezendonck-Goovaerts risk measures}
Having the Orlicz-Lorentz premia at our disposal, we can now define the distortion Haezendonck-Goovaerts risk measures by the same simple procedure as in Section \ref{s-HaezGoov}.

\begin{definition}\label{d-dHG}
Let $g$ be a distortion function, $\phi$ a normalized Young function, and $\alpha\in [0,1)$. The \textit{distortion Haezendonck-Goovaerts risk measure} $\rho_{g,\phi,\alpha}:L^\phi_g\to\R$ is given by
\[
\rho_{g,\phi,\alpha}(X) = \inf_{x\in\R} \big(\pi_{g,\phi,\alpha}((X-x)^+)+x\big).
\]
\end{definition}

It follows as in our discussion after Definition \ref{d-HaezGoov}, using Proposition \ref{p-distOrlicz2}(b), that we can assume here without loss of generality that $\phi$ is normalized and that we need to impose that $\alpha\geq 0$, and thus $\alpha\in [0,1)$, in order to have a risk measure.

\begin{remark}
The definition of the distortion Haezendonck-Goovaerts risk measure was suggested by Definition 4.2 of Goovaerts, Linders, Van Weert, and Tank \cite{GLVT12}, who call it the optimal generalized Haezendonck–Goovaerts risk measure; they consider the case when $\alpha\in (0,1)$. The link between the two definitions becomes clearer by noting that
\begin{align*}
&\rho_{g,\phi,\alpha}((X-x)^+)+x= \inf\Big\{a>x : \int_0^1 \phi\Big(\frac{(F_X^{-1}(1-u)-x)^+}{a-x}\Big)\d g(u)\leq 1-\alpha\Big\}.
\end{align*}
Thus the definitions coincide for $X\in L^\infty$ if $g$ is continuously differentiable with $g(0)=0$, see Proposition \ref{p-distOrlicz}(c).
\end{remark}

\begin{remark}\label{r-distHaezGoov}
 Let us convince ourselves that the distortion Haezendonck-Goovaerts risk measures are well defined. First, let $X\in L_g^\phi$ and $x\in \R$. By a property of quantile functions and the convexity of $\phi$ we have, for any $a>0$,
\begin{align*}
\int_0^1 \phi\Big(\frac{F_{(X-x)^+}^{-1}(1-u)}{a+1}\Big) \d g(u) &= \int_0^1 \phi\Big(\frac{(F_X^{-1}(1-u)-x)^+}{a+1}\Big) \d g(u)\\
 &\leq \int_0^1 \phi\Big(\frac{(|F_X^{-1}(1-u)|+|x|)}{a+1}\Big) \d g(u)\\
  &\leq \frac{a}{a+1}\int_0^1 \phi\Big(\frac{|F_X^{-1}(1-u)|}{a}\Big) \d g(u) + \frac{1}{a+1}\phi(|x|), 
\end{align*}
which shows that $(X-x)^+\in (L_g^\phi)_+$, so that $\pi_{g,\phi,\alpha}$ can be applied. This argument is valid for any Young function $\phi$ and any $\alpha<1$.

Secondly, if $\phi$ is normalized and $\alpha \in[0,1)$ then $\phi^{-1}(1-\alpha)\leq 1$. Thus it follows with Proposition \ref{p-distOrlicz2}(a) that, for any $x\in\mathbb{R}$,
\begin{equation}\label{eq-dist}
\begin{split}
\pi_{g,\phi,\alpha}((X-x)^+)+x &\geq \frac{\rho_g((X-x)^+)}{\phi^{-1}(1-\alpha)} +x \geq \rho_g((X-x)^+) +x\\
& \geq \rho_g(X-x)+x = \rho_g(X),
\end{split}
\end{equation}
where we have used that $\rho_g$ is cash-invariant and monotonic. Thus $\rho_{g,\phi,\alpha}(X)\in\R$, as required from a risk measure.
\end{remark}

\begin{example}\label{ex-distHaezGoov}
(a) If $\phi$ is the identity and $\alpha=0$, then, for $X\in L_g^\phi=L_g$, $\pi_{g,\phi,0}((X-x)^+)+x = \rho_g((X-x)^+)+x = \int_0^1((F_X^{-1}(1-u)-x)^++x)\d g(u)$ decreases as $x$ decreases. Thus, by the dominated convergence theorem, $\rho_{g,\phi,0}(X) = \rho_g(X)$. This will be considerably generalized in Corollary \ref{c-dHGtriv} below.

(b) Let $\alpha\in[0,1)$. If $g$ is the identity function then $\rho_{g,\phi,\alpha}(X) = \rho_{\phi,\alpha}(X)$ for all $X\in L_g^\phi=L^\phi$. 

(c) Let $\alpha\in[0,1)$ and $\beta\in (0,1)$. If $g(u)=\mathds{1}_{[1-\beta,1]}(u)$, then $\pi_{g,\phi,\alpha}(X) = \frac{\VaR_\beta(X)}{\phi^{-1}(1-\alpha)}$ for all risks $X\geq 0$. Since $\VaR_\beta((X-x)^+)=(\VaR_\beta(X)-x)^+$, we obtain that $\rho_{g,\phi,\alpha}(X)=\VaR_\beta(X)$ for any risk $X$, independently of $\alpha$.
\end{example}

\subsection{Distortion HG: Convex cone}\label{subs-convex}

From Example \ref{ex-cone} we know that the set $L_g^\phi$ of risks is not necessarily a convex cone, even if $g$ is concave. On the other hand, if $g$ is the identity then $L_g^\phi=L^\phi$ is (even) a vector space. By Theorem \ref{t-distOrlicz} we have that $(L_g^\phi)_+$ is a convex cone whenever $g$ is concave. This remains true for $L_g^\phi$ under additional assumptions on $g$.

\begin{proposition}\label{p-conv}
If $g$ is concave and constant on some interval $[u_0,1]$, $u_0\in[0,1)$, then $L_g^\phi$ is a convex cone.
\end{proposition}

\begin{proof}
We first claim that
\begin{equation}\label{eq-Lplus}
X\in L_g^\phi \Longleftrightarrow X^+\in (L_g^\phi)_+.
\end{equation}
To see this, let $u_1=P(X>0)$. Then $F_X^{-1}(1-u) > 0$ for $u < u_1$ and $F_X^{-1}(1-u)\leq 0$ for $u\geq u_1$. If $u_1=1$, $X\geq 0$, and the claim holds. Otherwise, we can assume that $u_1\leq u_0$. Then, for any $a>0$,
\begin{align}\label{eq-int}
\int_0^1 \phi\Big(\frac{|F_X^{-1}(1-u)|}{a}\Big) \d g(u) = \int_{[0,u_1)} \phi\Big(&\frac{F_X^{-1}(1-u)}{a}\Big) \d g(u)+\int_{[u_1,u_0]} \phi\Big(\frac{-F_X^{-1}(1-u)}{a}\Big) \d g(u)\notag \\
& +\int_{(u_0,1]} \phi\Big(\frac{-F_X^{-1}(1-u)}{a}\Big) \d g(u).
\end{align}
Here, the second term on the right is finite, and the third term vanishes by hypothesis.

Since $F_{X^+}^{-1}= (F_{X}^{-1})^+$, we have by the same argument that
\begin{align}\label{eq-intbis}
\int_0^1 \phi\Big(\frac{F_{X^+}^{-1}(1-u)}{a}\Big) \d g(u)= \int_{[0,u_1)} \phi\Big(\frac{F_{X}^{-1}(1-u)}{a}\Big) \d g(u).
\end{align}

Thus $\int_0^1 \phi\big(\frac{|F_X^{-1}(1-u)|}{a}\big) \d g(u)<\infty$ if and only if 
$\int_0^1 \phi\big(\frac{F_{X^+}^{-1}(1-u)}{a}\big) \d g(u)<\infty$, which proves the claim.

Let us now show that $L_g^\phi$ is a convex cone. Since it is invariant under positive scalar multiplication, we need to show that it is invariant under taking sums. Thus let $X,Y\in L_g^\phi$. By \eqref{eq-Lplus}, $X^+,Y^+\in (L_g^\phi)_+$, hence $X^++Y^+\in (L_g^\phi)_+$ by Theorem \ref{t-distOrlicz}. Since $(X+Y)^+\leq X^++Y^+$, also $(X+Y)^+\in (L_g^\phi)_+$ by Proposition \ref{p-lg}, and hence $X+Y\in L_g^\phi$ by \eqref{eq-Lplus} again.
\end{proof}

\begin{proposition}\label{p-convbis}
If $g$ is concave with $g'(1)>0$, then $L_g^\phi$ is a convex cone.
\end{proposition}

\begin{proof}
We first claim that
\begin{equation*}
X\in L_g^\phi \Longleftrightarrow X^+\in (L_g^\phi)_+ \text{ and } X^-\in L^\phi.
\end{equation*}
To see this, we first note that, since $g$ is concave, it is almost everywhere differentiable, and it is \text{(left-)}differentiable at $1$ with $g'(1)\geq 0$. We assume then that $g'(1)>0$. Let again $u_1=P(X>0)$, where we can once more assume that $u_1<1$. Choose any $u_0\in (u_1,1)$ where $g'$ is differentiable. We then have again \eqref{eq-int} and \eqref{eq-intbis}. 

This time, concerning the third term on the right of \eqref{eq-int}, we have for $u_0\leq u\leq 1$ that $c\leq g'(u)\leq d$ almost everywhere, where $d:=g'(u_0)$ and $c:=g'(1)>0$. Noting that $\d g(u)=g'(u)\d u$ on $[u_0,1]$, we thus see that the third term is finite if and only if 
$\int_{u_0}^1 \phi\big(\frac{-F_X^{-1}(1-u)}{a}\big)\d u$ is finite, hence if and only if
\begin{align*}
\int_{u_1}^1 \phi\Big(\frac{-F_X^{-1}(1-u)}{a}\Big)\d u=\int_{0}^1 \phi\Big(\frac{|F_{-X^-}^{-1}(1-u)|}{a}\Big)\d u=E\Big( \phi\Big(\frac{|X^-|}{a}\Big)\Big)
\end{align*}
is finite. 

Altogether we have that $\int_0^1 \phi\big(\frac{|F_X^{-1}(1-u)|}{a}\big) \d g(u)<\infty$ if and only if 
$\int_0^1 \phi\big(\frac{F_{X^+}^{-1}(1-u)}{a}\big) \d g(u)<\infty$ and $E\big(\phi\big(\frac{|X^-|}{a}\big)\big)<\infty$, which proves the claim.

From here, the proof can be finished as that of Proposition \ref{p-conv}, using that $(L^\phi_g)_+$ and $L^\phi$ are convex cones.
\end{proof}

\subsection{Distortion HG: the infimum}
The definition of the distortion Haezendock-Goovaerts risk measure as an infimum raises again the question whether it is, in fact, a minimum. If $\alpha\neq 0$, this is indeed the case, and we give conditions under which the minimum is unique. We first show the following, which is valid for any Young function and any $\alpha<1$.

\begin{proposition}\label{p-distHaezGoovconv}
Let $\phi$ be a Young function, $\alpha<1$, and $X\in L_g^\phi$. 

\emph{(a)} Then the mapping $x\mapsto \pi_{g,\phi,\alpha}((X-x)^+)$ is convex on $\mathbb{R}$.

\emph{(b)} Let $g$ be continuous, $g(0)=0$, $g>0$ on $(0,1]$, and let $\phi$ be strictly convex and satisfy the $\Delta_2$-condition. If $P(X=\esssup X)=0$ then $x\mapsto \pi_{g,\phi,\alpha}((X-x)^+)+x$ is strictly convex for $x < \esssup X$.
\end{proposition}

\begin{proof}
(a) Note that the functions $z\mapsto (z-x)^+$ are convex and increasing on $\mathbb{R}$ for any $x\in \mathbb{R}$. 

Now let $x,y\in \mathbb{R}$ and $0\leq \lambda\leq 1$. It follows that the risks $\lambda(X-x)^+$ and $(1-\lambda)(X-y)^+$ are comonotonic. Propositions \ref{p-distOrliczcom} and \ref{p-distOrliczbis} then imply that
\begin{align*}
\pi_{g,\phi,\alpha}((X-(\lambda x+(1-\lambda)y))^+)&= \pi_{g,\phi,\alpha}((\lambda(X-x)+(1-\lambda)(X-y))^+)\\
&\leq  \pi_{g,\phi,\alpha}(\lambda(X-x)^++(1-\lambda)(X-y)^+)\\
&\leq  \pi_{g,\phi,\alpha}(\lambda(X-x)^+)+\pi_{g,\phi,\alpha}((1-\lambda)(X-y)^+)\\
&=  \lambda \pi_{g,\phi,\alpha}((X-x)^+)+(1-\lambda)\pi_{g,\phi,\alpha}((X-y)^+),
\end{align*}
which had to be shown.

(b) Let $\psi=\frac{1}{1-\alpha}\phi$ and $\mu_g$ the measure induced by $g$, see the discussion after Definition \ref{d-distfunc}. Then, for any $Y\in (L_g^\phi)_+$,
\[
\pi_{g,\phi,\alpha}(Y)= \|F_Y^{-1}(1-\cdot)\|,
\]
where $\|\cdot\|$ is the Luxemburg norm in the Orlicz space $L^\psi([0,1],\mathcal{B}[0,1], \mu_g)$, see \cite{Che96}, \cite{EdSo92}, \cite[p.\ 54]{RaRe91}. Since $g$ is continuous with $g(0)=0$, the measure $\mu_g$ is atomless. Thus, under the stated additional assumptions, the above Luxemburg norm is rotund, see \cite[Section 7.1, Corollary 5]{RaRe91}, that is, for $Y_1,Y_2\in L^\psi$ not collinear and $0<\lambda<1$, $\|\lambda Y_1+(1-\lambda)Y_2\|<\lambda\|Y_1\|+(1-\lambda)\|Y_2\|$, see \cite[Proposition 5.1.11]{Meg98}. 

Now, let $x_1<x_2<\esssup X$. Then $F_{(X-x_1)^+}^{-1}(1-\cdot)= (F_{X}^{-1}-x_1)^+(1-\cdot)$ and $F_{(X-x_2)^+}^{-1}(1-\cdot)= (F_{X}^{-1}-x_2)^+(1-\cdot)$ are not collinear. Otherwise there were $a,b\in\mathbb{R}$ not both zero such that $a(F_{X}^{-1}-x_1)^+(1-\cdot)=b(F_{X}^{-1}-x_2)^+(1-\cdot)$ $\mu_g$-almost everywhere. Let $p=F_X(x_2)$. Since $x_2<\esssup X$, we have that $p<1$ and $F_X^{-1}(1-u)>x_2$ for $0\leq u<1-p$. Thus $a(F_{X}^{-1}-x_1)(1-u)=b(F_{X}^{-1}-x_2)(1-u)$ for $\mu_g$-almost every $u\in [0,1-p)$. By hypothesis, $\mu_g([0,q))>0$ for all $q>0$. Note that $a\neq b$ because otherwise $\mu_g(x_1=x_2)\geq \mu_g([0,1-p))>0$. Hence there is some $c\in\mathbb{R}$ such that $F_X^{-1}(1-u)=c$ for $\mu_g$-almost every $u\in [0,1-p)$; since $u\mapsto F_X^{-1}(1-u)$ is decreasing, we deduce that there is some $u_0>0$ such that 
$F_X^{-1}(1-u)=c$ for $0\leq u\leq u_0$, which implies that $c=\esssup X$ and $P(X=c)>0$, contradicting the hypothesis.

Now, using the convexity of the functions $z\mapsto (z-x)^+$, Proposition \ref{p-distOrliczbis}, and the comonotonic additivity of VaR, we have for $0<\lambda<1$ that
\begin{align*}
\pi_{g,\phi,\alpha}((X-(\lambda x_1+(1-\lambda) x_2))^+)&\leq \pi_{g,\phi,\alpha}(\lambda (X-x_1)^++ (1-\lambda)(X-x_2)^+)\\
&= \|F^{-1}_{\lambda (X-x_1)^++ (1-\lambda)(X-x_2)^+}(1-\cdot)\|\\
&= \|\lambda F^{-1}_{(X-x_1)^+}(1-\cdot)+ (1-\lambda)F^{-1}_{(X-x_2)^+}(1-\cdot)\|\\
&< \lambda\|F^{-1}_{(X-x_1)^+}(1-\cdot)\|+ (1-\lambda)\|F^{-1}_{(X-x_2)^+}(1-\cdot)\|\\
&= \lambda\pi_{g,\phi,\alpha}((X-x_1)^+)+ (1-\lambda)\pi_{\phi,\alpha}((X-x_2)^+),
\end{align*}
so that $x\mapsto \pi_{g,\phi,\alpha}((X-x)^+)+x$ is strictly convex for $x<\esssup X$.
\end{proof} 

\begin{example}\label{ex-strconv}
Let $g$ be the identity, $\phi(t)=t^2$, and $\alpha<1$. If $P(X=0)=P(X=1)=\frac{1}{2}$, then $\pi_{g,\phi,\alpha}((X-x)^+)+x= \frac{1}{\sqrt{2(1-\alpha)}}(1-x)+x$ for $0<x<1$, which is not strictly convex. Thus, part (b) of the proposition may fail if $P(X=\esssup X)>0$; the example also contradicts \cite[Proposition 11(f)]{BeRo08} and \cite[Proposition 3(c)]{BeRo12}.
\end{example}

Part (a) of the proposition implies that the minimum in the definition of the distortion Haezendonck-Goovaerts risk measure is attained if $\alpha\neq 0$.

\begin{proposition}\label{p-infmin}
Let $0<\alpha<1$ and $X\in L_g^\phi$. 

\emph{(a)} Then
\[
\rho_{g,\phi,\alpha}(X) = \min_{x\in\R} \big(\pi_{g,\phi,\alpha}((X-x)^+)+x\big).
\]

\emph{(b)} Let $g$ be continuous, $g(0)=0$, $g>0$ on $(0,1]$, and let $\phi$ be strictly convex and satisfy the $\Delta_2$-condition. If $P(X=\esssup X)=0$ then there is a unique value $x\in\mathbb{R}$ such that
\[
\rho_{g,\phi,\alpha}(X) =\pi_{g,\phi,\alpha}((X-x)^+)+x.
\]
\end{proposition}

\begin{proof} (a) We follow the proof of \cite[Proposition 3(b)]{BeRo12}. By Proposition \ref{p-distOrlicz2}(a) we have, for any $x\in\mathbb{R}$,
\[
\pi_{g,\phi,\alpha}((X-x)^+)+x \geq \frac{\rho_g((X-x)^+)}{\phi^{-1}(1-\alpha)} +x,
\]
and therefore by monotonicity and cash-invariance of $\rho_g$,
\begin{equation}\label{eq-dist3}
\pi_{g,\phi,\alpha}((X-x)^+)+x \geq \frac{\rho_g(X)-x}{\phi^{-1}(1-\alpha)} +x=\frac{\rho_g(X)}{\phi^{-1}(1-\alpha)}+x\Big(1-\frac{1}{\phi^{-1}(1-\alpha)}\Big).
\end{equation}
It follows from these two inequalities that the function $x\mapsto \pi_{g,\phi,\alpha}((X-x)^+)+x$ tends to $\infty$ as $x\to\pm\infty$; note that $\phi^{-1}(1-\alpha)<1$. Since the function is convex by Proposition \ref{p-distHaezGoovconv}(a), the result follows.

(b) This is a direct consequence of part (a), Proposition \ref{p-distHaezGoovconv}(b), and the fact that $\pi_{g,\phi,\alpha}((X-x)^+)+x=x$ for $x\geq\esssup X$.
\end{proof}

\begin{example}
A variant of Example \ref{ex-strconv} shows that part (b) of the proposition may fail for any $\alpha\in (0,1)$, if $P(X=\esssup X)>0$. Indeed, if $g$ is the identity, $\phi(t)=t^2$, $P(X=0)=\alpha$, and $P(X=1)=1-\alpha$, then $\pi_{g,\phi,\alpha}((X-x)^+)+x= 1$ for $0\leq x\leq 1$, so that the minimum is not uniquely attained. 
\end{example}

The proof of Proposition \ref{p-infmin} also gives us some information on the location of a minimum.

\begin{lemma}\label{l-estim}
Let $0<\alpha<1$. Let $Y_1, Y_2\in L_g^\phi$ and $Y_1\leq X\leq Y_2$. If
\[
\rho_{g,\phi,\alpha}(X)=\pi_{g,\phi,\alpha}((X-x)^+)+x
\]
then
\[
\frac{\rho_g(Y_1)-\phi^{-1}(1-\alpha)\rho_{g,\phi,\alpha}(Y_2)}{1-\phi^{-1}(1-\alpha)}\leq x\leq \rho_{g,\phi,\alpha}(Y_2).
\]
\end{lemma}

\begin{proof}
First note that, by Proposition \ref{p-inv}, $X\in L_g^\phi$. The right-hand inequality is clear by positivity of $\pi_{g,\phi,\alpha}$ and monotonicity of $\rho_{g,\phi,\alpha}$. Next, by \eqref{eq-dist3},
\[
x\geq \frac{\frac{\rho_g(X)}{\phi^{-1}(1-\alpha)}-\rho_{g,\phi,\alpha}(X)}{\frac{1}{\phi^{-1}(1-\alpha)}-1},
\]
which implies the left-hand inequality by the monotonicity of $\rho_g$ and $\rho_{g,\phi,\alpha}$.
\end{proof}

Of course, one obtains the best estimate if $Y_1=Y_2=X$, but it is in the above form that the lemma will be useful in the sequel.

For $\alpha=0$, the situation is quite different.

\begin{proposition}\label{p-distHaezGoovinf}
Let $\alpha=0$ and $X\in L_g^\phi$.

\emph{(a)} Then $x\mapsto \pi_{g,\phi,0}((X-x)^+)+x$ is increasing on $\mathbb{R}$. In particular,
\[
\rho_{g,\phi,0}(X) = \lim_{x\to-\infty} \big(\pi_{g,\phi,0}((X-x)^+)+x\big).
\]

\emph{(b)} Let $g$ be continuous, $g(0)=0$, $g>0$ on $(0,1]$, and let $\phi$ be strictly convex and satisfy the $\Delta_2$-condition. If $P(X=\esssup X)=0$ then $x\mapsto \pi_{g,\phi,0}((X-x)^+)+x$ is strictly increasing on $\mathbb{R}$. In particular, the function does not attain its infimum.
\end{proposition}

\begin{proof}
(a) Let $x_1<x_2$. Using Proposition \ref{p-distOrliczcom}, applied to the comonotonic risks $(X-x_1)^+-(X-x_2)^+$ and $(X-x_2)^+$, the fact that $(x-x_1)^+-(x-x_2)^+\leq x_2-x_1$ for all $x\in\mathbb{R}$, and Propositions \ref{p-distOrliczbis} and \ref{p-distOrlicz2}(b), we obtain that
\begin{align*}
\pi_{g,\phi,0}((X-x_1)^+)&= \pi_{g,\phi,0}((X-x_1)^+-(X-x_2)^++(X-x_2)^+)\\
&\leq \pi_{g,\phi,0}((X-x_1)^+-(X-x_2)^+)+\pi_{g,\phi,0}((X-x_2)^+)\\
&\leq x_2-x_1 + \pi_{g,\phi,0}((X-x_2)^+),
\end{align*}
which implies the claim.

(b) This is a direct consequence of part (a), Proposition \ref{p-distHaezGoovconv}(b), and the fact that $\pi_{g,\phi,0}((X-x)^+)+x=x$ for $x\geq\esssup X$.
\end{proof} 

As in the undistorted case, for $\alpha=0$ the distortion Haezendonck-Goovaerts risk measure often reduces to the corresponding distortion risk measure. Since we first need some more knowledge about these risk measures, we postpone the discussion, see Theorem \ref{t-dHGtriv} below.

\subsection{Distortion HG: risk theoretic properties}
We collect several important properties of the distortion Haezendonck-Goovaerts risk measures. In these results, when we do not restrict $\alpha$, we suppose that $\alpha\in[0,1)$.

\begin{proposition}\label{p-distHaezGoov}
Let $X\in L_g^\phi$. Then:
\begin{enumerate}[label=\emph{(\alph*)}] 
\item $\rho_{g,\phi,\alpha}(X) \leq \pi_{g,\phi,\alpha}(X^+)$.
\item $\rho_g(X)\leq \rho_{g,\phi,\alpha}(X) \leq \esssup X$.
\end{enumerate}
Suppose, in addition, that $g:[0,1]\to [0,1]$ is bijective, and let $\alpha\neq 0$. Then:
\begin{enumerate}[label=\emph{(\alph*)},resume]
\item $\rho_{g,\phi,\alpha}(X) \geq \VaR_{1-g^{-1}(1-\alpha)}(X)$.
\end{enumerate}
\end{proposition}

\begin{proof}
(a) is obvious by taking $x=0$ in the definition of $\rho_{g,\phi,\alpha}$.

(b) The first inequality follows form \eqref{eq-dist} by definition of $\rho_{g,\phi,\alpha}$. The second inequality is trivial if $\esssup X=\infty$; otherwise it follows by taking $x=\esssup X$ in the definition of $\rho_{g,\phi,\alpha}$.

(c) We note that, for any $a>0$ and $b\in\R$, $\mathds{1}_{\{ b>a\}}\leq \phi\big(\frac{b^+}{a}\big)$. Thus, for any $a>0$ and $x\in \R$,
\begin{align*}
\int_0^1 \phi\Big(\frac{F_{(X-x)^+}^{-1}(1-u)}{a}\Big) \d g(u)&= \int_0^1 \phi\Big(\frac{(F_X^{-1}(1-u)-x)^+}{a}\Big) \d g(u)\\
&\geq  \int_0^1 \mathds{1}_{\{ F_X^{-1}(1-u)-x>a\}} \d g(u)=  \int_0^1 \mathds{1}_{\{ u< 1- F_X(a+x)\}} \d g(u)\\
&= g(1- F_X(x+a)),
\end{align*}
where we have applied properties of quantile functions; note also that $g$ is necessarily continuous with $g(0)=0$. Hence $\pi_{g,\phi,\alpha}((X-x)^+) \geq \inf\{ a : g(1- F_X(a+x))\leq 1-\alpha\}= \inf\{ a : F_X(a+x)\geq 1-g^{-1}(1-\alpha) \} = \VaR_{1-g^{-1}(1-\alpha)}(X)-x$. The definition of $\rho_{g,\phi,\alpha}$ then yields the claim.
\end{proof}

\begin{proposition}\label{p-distHaezGoov2}
The distortion Haezendonck-Goovaerts risk measure $\rho_{g,\phi,\alpha}$ is monotonic, cash-invariant and positively homogeneous on $L_g^\phi$.
\end{proposition}

\begin{proof}
Cash-invariance follows from the identity
\begin{align*}
\pi_{g,\phi,\alpha}((X+b-x)^+)+x = \pi_{g,\phi,\alpha}((X-(x-b))^+)+(x-b) +b.
\end{align*} 
Monotonicity passes from $\pi_{g,\phi,\alpha}$ to $\rho_{g,\phi,\alpha}$ since $(X-x)^+\leq (Y-x)^+$ if $X\leq Y$. Positive homogeneity for $\lambda>0$ follows from the identity
\[
\pi_{g,\phi,\alpha}((\lambda X-x)^+)+x =\lambda\big(\pi_{g,\phi,\alpha}((X-\tfrac{x}{\lambda})^+)+\tfrac{x}{\lambda}\big).
\] 
For $\lambda=0$ we note that $\pi_{g,\phi,\alpha}((0-0)^+)+0=0$, $\pi_{g,\phi,\alpha}((0-x)^+)+x\geq 0$ if $x>0$, and $\pi_{g,\phi,\alpha}((0-x)^+)+x=\pi_{g,\phi,\alpha}(-x)+x=(-x)\pi_{g,\phi,\alpha}(1)+x\geq 0$ if $x<0$, where we have used the positive homogeneity of $\pi_{g,\phi,\alpha}$ and that $\pi_{g,\phi,\alpha}(1)\geq 1$ by Proposition \ref{p-distOrlicz2}(b). Thus, $\rho_{g,\phi,\alpha}(0)=0$.
\end{proof}

The distortion Haezendonck-Goovaerts risk measures are subadditive for comonotonic risks.

\begin{proposition}\label{p-distHGcom}
Let $X,Y\in L_g^\phi$ be comonotonic risks. Then $X+Y\in L_g^\phi$ and
\[
\rho_{g,\phi,\alpha}(X+Y)\leq \rho_{g,\phi,\alpha}(X)+\rho_{g,\phi,\alpha}(Y).
\]
\end{proposition}

\begin{proof}
Let $X,Y\in L_g^\phi$ be comonotonic. Since VaR is additive for comonotonic risks, we have that $|F_{X+Y}^{-1}(1-u)|\leq |F_{X}^{-1}(1-u)|+|F_{Y}^{-1}(1-u)|$ for all $u\in[0,1)$. Thus, by the argument in the proof of Proposition \ref{p-distOrliczcom} and using the monotonicity of $\phi$, we find that $X+Y\in L_g^\phi$.

Next, let $x,y\in \mathbb{R}$. Since $X$ and $Y$ are comonotonic, there is a random variable $Z$ with values in an interval $I\subset \R$ and two increasing functions $f_1,f_2:I\to\R$ such that $(X,Y)$ and $(f_1(Z),f_2(Z))$ have the same distribution. But then $((X-x)^+,(Y-y)^+)$ and $((f_1(Z)-x)^+,(f_2(Z)-y)^+)$ have the same distribution, so that also $(X-x)^+$ and $(Y-y)^+$ are comonotonic. Thus, by Proposition \ref{p-distOrliczcom} and the monotonicity of $\pi_{g,\phi,\alpha}$, we obtain that
\begin{align*}
\pi_{g,\phi,\alpha}((X+Y-(x+y))^+)+(x+y)&\leq \pi_{g,\phi,\alpha}((X-x)^++(Y-y)^+)+(x+y)\\
&\leq \pi_{g,\phi,\alpha}((X-x)^+)+x+\pi_{g,\phi,\alpha}((Y-y)^+)+y.
\end{align*}
Taking infima on both sides implies the claim.
\end{proof}

We turn to continuity properties. In the following results, some proofs require that $\alpha\neq 0$.

\begin{proposition}\label{p-distHGFatou}
If $0<\alpha<1$, then $\rho_{g,\phi,\alpha}$ has the Fatou property on $L_g^\phi$.
\end{proposition}

\begin{proof}
By Remark \ref{r-Fatou}(c) and Propositions \ref{p-inv} and \ref{p-distHaezGoov2} it suffices to show that if $X_n\nearrow X$ and $X_1, X\in L_g^\phi$ then $\rho_{g,\phi,\alpha}(X)\leq \limsup_{n\to\infty}\rho_{g,\phi,\alpha}(X_n)$.

For this, we use an idea of \cite{GMX20}. By Proposition \ref{p-infmin}, for any $n$, there are $x_n\in\mathbb{R}$ such that $\rho_{g,\phi,\alpha}(X_n)=\pi_{g,\phi,\alpha}((X_n-x_n)^+)+x_n$. Since $X_1\leq X_n\leq X$ for all $n$, it follows from Lemma \ref{l-estim} that the sequence $(x_n)_n$ is bounded, hence has a convergent subsequence. We may then assume that the whole sequence converges, and we put $x_0=\lim_{n\to\infty} x_n$. But then $(X_n-x_n)^+\to (X-x_0)^+$ and $0\leq (X_n-x_n)^+\leq (X-\inf_k|x_k|)^+\in (L_g^\phi)^+$ for all $n$. Using Proposition \ref{p-distOrliczFatou}, we then get that
\begin{align*}
\rho_{g,\phi,\alpha}(X) &\leq \pi_{g,\phi,\alpha}((X-x_0)^+) +x_0\leq \liminf_{n\to\infty}\pi_{g,\phi,\alpha}((X_n-x_n)^+) +x_0\\
&=\liminf_{n\to\infty}\big(\pi_{g,\phi,\alpha}((X_n-x_n)^+) +x_n\big)=\liminf_{n\to\infty}\rho_{g,\phi,\alpha}(X_n)=\limsup_{n\to\infty}\rho_{g,\phi,\alpha}(X_n),
\end{align*}
as desired.
\end{proof}

For the reverse Fatou property, recall the Property ($P_{g,\phi}$) stated before Proposition \ref{p-OLrevFatou}. The result holds for all $\alpha\in[0,1)$.

\begin{proposition}\label{p-distHGFatourev}
If \emph{($P_{g,\phi}$)} holds, then $\rho_{g,\phi,\alpha}$ has the reverse Fatou property on $L_g^\phi$.
\end{proposition}

\begin{proof}
It suffices by Remark \ref{r-Fatou}(c) and Propositions \ref{p-inv} and \ref{p-distHaezGoov2} to show that if $X_n\searrow X$ and $X_1,X\in L_g^\phi$ then $\rho_{g,\phi,\alpha}(X)\geq \inf_{n\geq 1}\rho_{g,\phi,\alpha}(X_n)$. 

For this, we follow the proof of \cite[Proposition 17]{BeRo08}. By Proposition \ref{p-OLrevFatou} we have, for all $x\in\mathbb{R}$,
\[
\pi_{g,\phi,\alpha}((X-x)^+)\geq \inf_{n\geq 1}\pi_{g,\phi,\alpha}((X_n-x)^+).
\]
Hence
\begin{align*}
\inf_{n\geq 1}\rho_{g,\phi,\alpha}(X_n)&= \inf_{n\geq 1}\inf_{x\in\mathbb{R}}(\pi_{g,\phi,\alpha}((X_n-x)^+)+x)= \inf_{x\in\mathbb{R}}\inf_{n\geq 1}(\pi_{g,\phi,\alpha}((X_n-x)^+)+x) \\
&\leq \inf_{x\in\mathbb{R}}(\pi_{g,\phi,\alpha}((X-x)^+)+x)= \rho_{g,\phi,\alpha}(X),
\end{align*}
as desired.
\end{proof}

Using Proposition \ref{p-OLrevFatoubdd} instead of Proposition \ref{p-OLrevFatou}, we obtain in the same way a variant on $L^\infty$.

\begin{proposition}\label{p-distHGFatourevbdd} 
If $g(0)=0$ and $g$ is continuous, then $\rho_{g,\phi,\alpha}$ has the reverse Fatou property on $L^\infty$.
\end{proposition}

Unfortunately, we only have partial converses: we are not able to show that $g$ must be continuous if $\rho_{g,\phi,\alpha}$ has the reverse Fatou property. Using ideas from the proof of \cite[Proposition 3.4]{GMX20}, we have the following.

\begin{lemma}\label{l-rhopi}
Let $0<\alpha<1$. If $(X_n)_n$ is a decreasing sequence in $(L_g^\phi)_+$, then $\rho_{g,\phi,\alpha}(X_n)\to 0$ implies that $\pi_{g,\phi,\alpha}(X_n)\to 0$.
\end{lemma}

\begin{proof}
Let us define $\sigma_n(x)=\pi_{g,\phi,\alpha}((X_n-x)^+)+x$, $x\in\mathbb{R}$. By Proposition \ref{p-infmin}, there are $x_n\in\mathbb{R}$ such that $\rho_{g,\phi,\alpha}(X_n)=\sigma_n(x_n)$, $n\geq 1$. By Lemma \ref{l-estim}, applied with $Y_1=0$ and $Y_2=X_n$, $\rho_{g,\phi,\alpha}(X_n)\to 0$ implies that $x_n\to 0$. 

Now, the functions $\sigma_n$ are convex by Proposition \ref{p-distHaezGoovconv}(a). If $x_n>0$, then 
\begin{align*}
0\leq \sigma_n(0) &= \sigma_n(\tfrac{1}{1+x_n}x_n+\tfrac{x_n}{1+x_n}(-1))\leq \tfrac{1}{1+x_n}\sigma_n(x_n)+\tfrac{x_n}{1+x_n}\sigma_n(-1)\\
&\leq \tfrac{1}{1+x_n}\sigma_n(x_n)+\tfrac{x_n}{1+x_n}\sigma_1(-1),
\end{align*}
where in the last line we have used the monotonicity of $\pi_{g,\phi,\alpha}$. In the same way, if $x_n<0$, then
\[
0\leq \sigma_n(0) \leq \tfrac{1}{1-x_n}\sigma_n(x_n)+\tfrac{-x_n}{1-x_n}\sigma_1(1).
\]
Since $x_n\to 0$ and $\sigma_n(x_n)\to 0$, we have altogether that $\pi_{g,\phi,\alpha}(X_n)=\sigma_n(0)\to 0$.
\end{proof}

\begin{proposition}\label{p-distHGFatourevnec}
Let $0<\alpha<1$. If the underlying probability space $(\Omega,\mathcal{A},P)$ is atomless and if $\rho_{g,\phi,\alpha}$ has the reverse Fatou property on $L_g^\phi$ then $g(0)=0$, and if $g$ is continuous on some neighbourhood of $0$ then \textit{either} $g=0$ on some neighbourhood of $0$ \textit{or} $\phi$ satisfies the $\Delta_2$-condition.
\end{proposition}   
  
\begin{proof}
We first show that $g(0)=0$. To see this, let $(A_n)_n$ be a decreasing sequence of sets in $\mathcal{A}$ with $P(A_n)=\frac{1}{n}$; see the proof of Proposition \ref{p-distrevFatou}. If $X_n=\mathds{1}_{A_n}$, $n\geq 1$, then $X_n\searrow 0$; also, $X_n\in L_g^\phi$ as bounded risks. By the reverse Fatou property, we have that $\rho_{g,\phi,\alpha}(X_n)\to 0$. By Lemma \ref{l-rhopi}, $\pi_{g,\phi,\alpha}(X_n)\to 0$. Now, a simple calculation shows that
\[
\pi_{g,\phi,\alpha}(X_n)=\frac{1}{\phi^{-1}\big(\frac{1-\alpha}{g(\frac{1}{n}-)}\big)},
\]
see also the proof of Proposition \ref{p-OLrevFatou}. We then deduce that $g(0)=0$.

Next suppose that $g$ is continuous on some neighbourhood of $0$, $g>0$ on $(0,1]$, and that $\phi$ does not satisfy the $\Delta_2$-condition. Then, by Lemma \ref{l-conv}(b), there are risks $X_n\in (L_g^\phi)_+$ such that $X_n\searrow 0$ with $\pi_{g,\phi,\alpha}(X_n)\geq \frac{1}{2}$ for all $n$. It follows from Lemma \ref{l-rhopi} that $\rho_{g,\phi,\alpha}(X_n)\not\to 0$, contradicting the reverse Fatou property.
\end{proof}

The above proof also gives a version on $L^\infty$.

\begin{proposition} 
Let $0<\alpha<1$. If the underlying probability space $(\Omega,\mathcal{A},P)$ is atomless and if $\rho_{g,\phi,\alpha}$ has the reverse Fatou property on $L^\infty$ then $g(0)=0$.
\end{proposition} 

Proposition \ref{p-st} easily implies the following. Indeed, $X\leq_{\text{st}} Y$ implies that $(X-x)^+\leq_{\text{st}}(Y-x)^+$ for any $x$; it suffices to note that $F^{-1}_{(X-x)^+}=(F^{-1}_{X}-x)^+$.

\begin{proposition}\label{p-stdHG}
Let $\alpha\in[0,1)$ and $X,Y\in L_g^\phi$. Then 
\[
X\leq_{\emph{\text{st}}} Y \Longrightarrow \rho_{g,\phi,\alpha}(X)\leq \rho_{g,\phi,\alpha}(Y).
\]
\end{proposition}

\subsection{Distortion HG: the concave case}

First, Proposition \ref{p-distOrliczsl} easily yields the following. It suffices to note that if $\varphi$ is an increasing convex function, then so is $z\mapsto \varphi((z-x)^+)$, hence $X\leq_{\text{sl}} Y$ implies that $(X-x)^+\leq_{\text{sl}}(Y-x)^+$ for any $x$. 

\begin{proposition}\label{p-distHGsl}
Let $g$ be concave. If $X,Y\in L_g^\phi$, then
\[
X\leq_{\emph{\text{sl}}} Y \Longrightarrow \rho_{g,\phi,\alpha}(X)\leq \rho_{g,\phi,\alpha}(Y).
\]
\end{proposition}

We also need a variant of Lemma \ref{l-risksub}, which is proved quite similarly.

\begin{lemma}\label{l-risksubbis}
Suppose that the underlying probability space $(\Omega,\mathcal{A},P)$ is atomless. Let $\mathcal{X}$ be a set of risks on $\Omega$ that contains the constants and $\rho:\mathcal{X}\to\mathbb{R}$ a risk measure such that
\begin{enumerate}[label=\emph{(\roman*)}] 
\item $X$ a risk, $Y\in \mathcal{X}$, $X =_d Y$ $\Longrightarrow$ $X\in \mathcal{X}$, $\rho(X)= \rho(Y)$;
\item $X,Y\in \mathcal{X}$ comonotonic $\Longrightarrow$ $X+Y\in\mathcal{X}$ and $\rho(X+Y)\leq \rho(X)+\rho(Y)$;
\item $X,Y\in \mathcal{X}$, $X\leq_{\emph{\text{sl}}} Y$ $\Longrightarrow$ $\rho(X)\leq \rho(Y)$.
\end{enumerate}
Then, for all $X,Y\in\mathcal{X}$, if $X+Y\in\mathcal{X}$ then $\rho(X+Y)\leq \rho(X)+\rho(Y)$; that is, $\rho$ is subadditive.
\end{lemma}

We arrive at the main result of this paper. It follows, as in the proof of Theorem \ref{t-distOrlicz}, from Lemma \ref{l-risksubbis} and Propositions \ref{p-distHaezGoov2}, \ref{p-distHGcom} and \ref{p-distHGsl}.

\begin{theorem}\label{t-distHaezGoov2}
Let $g$ be concave. Then the distortion Haezendonck-Goovaerts risk measure $\rho_{g,\phi,\alpha}$ is coherent on $L_g^\phi$.
\end{theorem}

The proof by the first author given in \cite{Gou22} was based on Theorem \ref{t-distOrlicz}, using a generalization of \cite[Proposition 13]{BeRo08} and a variant of \cite[Theorem 1]{Roc74}.

We recall that, by Example \ref{ex-cone}, the set $L_g^\phi$ is not necessarily a convex cone, even if $g$ is concave and $\phi$ is the identity.

We do not know if concavity of $g$ is necessary for the coherence of $\rho_{g,\phi,\alpha}$, see Problem \ref{pr-coh}.

\subsection{The case of $\alpha=0$}
We turn to the announced reduction of the distortion Haezendonck-Goovaerts risk measure $\rho_{g,\phi,0}$.

\begin{theorem}\label{t-dHGtriv}
Let $\alpha=0$. Then, for all $X\in L^\infty$,
\[
\rho_g(X)\leq \rho_{g,\phi,0}(X)\leq \frac{c_+}{c_-}\rho_g(X^+) + \frac{c_-}{c_+}\rho_g(-X^-),
\]
where $c_-$ is the left derivative of $\phi$ at $1$, and $c_+$ is the right derivative of $\phi$ at $1$. If $\phi$ satisfies the $\Delta_2$-condition, then this holds for all $X\in L_g^\phi$.
\end{theorem}

\begin{corollary}\label{c-dHGtriv}
Let $\alpha=0$. If $\phi$ is differentiable at $1$ and satisfies the $\Delta_2$-condition, then
\[
\rho_{g,\phi,0}=\rho_g
\]
on $L_g^\phi$.
\end{corollary}

Since the proof is quite technical, we relegate it to the Appendix, see Section \ref{s-app}.

\section{Concluding remarks}\label{s-end}

\subsection{Problems} We suggest the following.

\begin{problem}\label{pr-coh}
Let $g$ be a distortion function, $\phi$ a normalized Young function, and $0<\alpha<1$. Characterize the coherence of the distortion Haezendonck-Goovaerts risk measure $\rho_{g,\phi,\alpha}$.
\end{problem}

\begin{problem}\label{pr-revFatou}
Let $g$ be a distortion function, $\phi$ a normalized Young function, and $0<\alpha<1$. Characterize the validity of the reverse Fatou property for the distortion Haezendonck-Goovaerts risk measure $\rho_{g,\phi,\alpha}$ on $L_g^\phi$.
\end{problem}

It might also be of interest, though of little consequence, to explore further the properties of $\rho_{g,\phi,\alpha}$ for $\alpha=0$. In particular, we propose the following.

\begin{problem}\label{pr-Fatou0}
Let $g$ be a distortion function, $\phi$ a normalized Young function, and $\alpha=0$. Does $\rho_{g,\phi,0}$ always have the Fatou property on $L_g^\phi$? Characterize the validity of the reverse Fatou property for $\rho_{g,\phi,0}$ on $L_g^\phi$.
\end{problem}

\subsection{Related work}\label{s-rel}
Wu and Xu \cite{WuXu22} have also, and independently, defined the Orlicz-Lorentz premium and the distortion Haezendonck-Goovaerts risk measure, but only for bounded risks and for distortion functions $g$ that are continuous and satisfy $g(0)=0$. More precisely, given a continuous increasing function $w:[0,1]\to [0,1]$ with $w(0)=0$ and $w(1)=1$, a strictly increasing normalized Young function, and $\alpha\in [0,1)$, they define a premium for $X\in L^\infty_+$ as
\[
\pi(X)= \inf\Big\{a>0 : \int_0^\infty \phi(x) \d (w\circ F_{X/a})(x)\leq 1-\alpha\Big\},
\]
see \cite[equation (1.6)]{WuXu22}. Now, considering the distortion function 
\begin{equation}\label{eq-gw}
g(u)= 1-w(1-u), \ u\in [0,1],
\end{equation}
and arguing as in the proof of Proposition \ref{p-equiv}(b), one sees that 
\[
\int_0^\infty \phi(t) \d (w\circ F_{X/a})(t)= \int_0^1 \phi\Big(\frac{F_X^{-1}(1-u)}{a}\Big) \d g(u),
\]
so that $\pi$ is the Orlicz-Lorentz premium $\pi_{g,\phi,\alpha}$; note that $g$ is continuous and $g(0)=0$.

Wu and Xu then define a risk measure for $X\in L^\infty$ in the usual way by
\[
\rho(X) = \inf_{x\in\R} \big(\pi((X-x)^+)+x\big),
\]
see \cite[equation (1.10)]{WuXu22}. In that context they obtain Propositions \ref{p-distOrlicz2}, \ref{p-distOrliczbis}, \ref{p-infmin} and Theorems \ref{t-distOrlicz}, \ref{t-distHaezGoov2}, see \cite[Propositions 2.1 and 4.1]{WuXu22}; their proof of coherence relies on the coherence of TVaR, see \cite[Appendix A]{WuXu22}. However, in \cite[Proposition 2.1(i)]{WuXu22} they claim that the infimum in the definition of $\pi$ is always attained if $X\neq 0$. Example \ref{ex-distOrlicz2} above shows that this is not the case (a fact also noted in \cite[p.\ 18]{ChRe25}).

Motivated by the paper of Wu and Xu, Chudziak and Rela \cite{ChRe25} have further generalized the Orlicz-Lorentz premia by replacing the function $g(\overline{F}_X(x))=g(P(X>x))$ in \eqref{eq-OL} \& \eqref{eq-defdist} by $\mu(\{X>x\})$ for a general capacity $\mu$, using Choquet integrals, see \cite[equations (3), (5), (6)]{ChRe25}. We remark, however, that their counter-example to \cite[Proposition 2.1(ix)]{WuXu22} in \cite[p.\ 19]{ChRe25} is not correct; they identify Wu and Xu's $w$ with $g$, while the correct link is given in \eqref{eq-gw} above, so that a convex $w$ in fact corresponds to a concave $g$.

\subsection{Robust versions}
Returning to our general situation, let $\pi_{g,\phi,\alpha}$ be the Orlicz-Lorentz premium for a distortion function $g$, a Young function $\phi$, and $\alpha< 1$. As we have noted in Remark \ref{rem-OL}(b), we interpret $g$ as the probability perception function and $\phi$ essentially as the utility function of the decision maker. However, identifying $g$ and $\phi$ will be difficult in practice. 

Faced with the same problem, Bellini, Laeven, and Rosazza Gianin \cite{BLR18} have introduced, in the case of Orlicz premia, a robust version for bounded positive risks. Assuming that the true but unknown $\phi$ belongs to a family $\mathcal{F}$ of Young functions, they define the \textit{robust Orlicz premium} as
\[
\pi_{\mathcal{F},\alpha}(X)  = \inf\Big\{a>0 : \sup_{\phi\in\mathcal{F}} E\Big(\phi\Big(\frac{X}{a}\Big)\Big)\leq 1-\alpha\Big\},\ X\in L^\infty_+,
\]
see \cite[Definition 4]{BLR18} (they take $\alpha=0$). The corresponding \textit{robust Haezendonck-Goovaerts risk measure} is then given, as usual, by
\[
\rho_{\mathcal{F},\alpha}(X)= \inf_{x\in\R} \big(\pi_{\mathcal{F},\alpha}((X-x)^+)+x\big),\ X\in L^\infty,
\]
see \cite[Section 5.1]{BLR18}.

In the same way, given a family $\mathcal{G}$ of distortion functions, Wang and Xu \cite{WaXu23} have introduced a \textit{(preference) robust distortion risk measure} by
\[
\pi_{\mathcal{G}}(X) = \sup_{g\in\mathcal{G}}\int_0^1 F_X^{-1}(1-u)\d g(u).
\]

It could thus be of interest to consider \textit{robust Orlicz-Lorentz premia}:
\begin{align*}
\pi_{g,\mathcal{F},\alpha}(X)= \inf\Big\{a>0 : \sup_{\phi\in\mathcal{F}}\int_0^1 \phi\Big(\frac{F_X^{-1}(1-u)}{a}\Big) \d g(u)\leq 1-\alpha\Big\},\\
\pi_{\mathcal{G},\phi,\alpha}(X)= \inf\Big\{a>0 : \sup_{g\in\mathcal{G}}\int_0^1 \phi\Big(\frac{F_X^{-1}(1-u)}{a}\Big) \d g(u)\leq 1-\alpha\Big\}.
\end{align*}
From there, one obtains \textit{robust distorted Haezendonck-Goovaerts risk measures} $\rho_{g,\mathcal{F},\alpha}$ and $\rho_{\mathcal{G},\phi,\alpha}$ in the usual way. We will not pursue this here. 

The robust versions $\pi_{\mathcal{G},\phi,\alpha}$ and $\rho_{\mathcal{G},\phi,\alpha}$ have already been introduced by Wu and Xu \cite{WuXu22} under the names RDOP and RDHG.

\section{Appendix}\label{s-app}
We first prove claims made in Remark \ref{rem-Lorentz} concerning the relationship between the domain $L_g$ of a distortion risk measure and the Lorentz spaces. For this, let $w:[0,1]\to\mathbb{R}$ be a positive measurable function with $\int_0^1 w(u)\d u=1$. Define $g(u)=\int_0^u w(v)\d v$, $u\in [0,1]$, which is a distortion function. Then consider the (classical) Lorentz space
\[
\Lambda(w)=\Big\{ X : \|X\|:=\int_0^1 F_{|X|}^{-1}(1-u)w(u) \d u <\infty\Big\}.
\]

\begin{proposition}\label{p-LorentzLgRhog}
We have that $X\in L_g$ if and only if $X^+\in \Lambda(w)$ and $\rho:=\inf_{x\in\mathbb{R}} (\|(X-x)^+\|+x)>-\infty$; in that case, $\rho_g(X)=\rho$.
\end{proposition}

\begin{proof}
For the proof of necessity, follow the argument in Example \ref{ex-distHaezGoov}(a) and note that $\rho_g(X)=\|X\|$ if $X\geq 0$. For sufficiency, let $x\leq 0$, and write $I_-=\{u\in (0,1): F_X^{-1}(1-u)\leq 0\}$, $I_+=(0,1)\setminus I_-$. Then
\begin{align*}
\|(X-x)^+\|+x &= \int_0^1(F^{-1}_{(X-x)^+}(1-u)+x)w(u)\d u\\
&= \Big(\int_{I_-}+\int_{I_+}\Big)\big((F^{-1}_{X}(1-u)-x)^++x\big)w(u)\d u.
\end{align*}
Since $x\leq 0$, the second integral coincides with
\begin{equation}\label{eq-lg}
\int_{I_+}F^{-1}_{X}(1-u)w(u)\d u = \int_{I_+}F^{-1}_{X^+}(1-u)w(u)\d u<\infty, 
\end{equation}
where we have used the first hypothesis. Thus, the second hypothesis implies that
\[
\inf_{x\leq 0} \int_{I_-}\big((F^{-1}_{X}(1-u)-x)^++x\big)w(u)\d u >-\infty.
\]
Since the integrands are negative and decrease as $x$ decreases, the monotone convergence theorem implies that $\int_{I_-} F^{-1}_{X}(1-u) w(u)\d u >-\infty$, hence 
\[
\int_{I_-} |F^{-1}_{X}(1-u)| w(u)\d u <\infty.
\]
Altogether we get that
\[
\int_0^1|F^{-1}_{X}(1-u)|w(u)\d u= \Big(\int_{I_-}+\int_{I_+}\Big)|F^{-1}_{X}(1-u)|w(u)\d u<\infty,
\]
where the second integral is finite by \eqref{eq-lg}.
\end{proof}

\begin{proposition}\label{p-LorentzLg}
If $w$ is decreasing, then $\Lambda(w)\subset L_g$.
\end{proposition}

\begin{proof}
Let $X\in \Lambda(w)$. We claim that $X\in L_g$, that is $\int_{0}^1 |F_{X}^{-1}(1-u)| w(u)\d u<\infty$.

First, by monotonicity of VaR, we have that $F_X^{-1}\leq F_{|X|}^{-1}$, hence $|F_X^{-1}(1-u)|\leq F_{|X|}^{-1}(1-u)$ if $F_X^{-1}(1-u)\geq 0$. Secondly, the upper and lower quantile functions coincide almost everywhere, see \cite[Lemma A.19]{FoSc16}. Thus, with \cite[equation (4.44)]{FoSc16} we have that $F_{X}^{-1}(1-u)=-F_{-X}^{-1}(u)\geq -F_{|X|}^{-1}(u)$ for almost all $u\in[0,1]$, and hence $|F_{X}^{-1}(1-u)|\leq F_{|X|}^{-1}(u)$ a.e. if $F_{X}^{-1}(1-u)\leq 0$.

Now, if $X\geq 0$, then there is nothing to prove. 

Next, let $X\leq 0$, so that $F_{X}^{-1}\leq 0$ on $[0,1]$. Since $w$ is decreasing, we have that $w(u)\leq w(1-u)$ if $u\geq \frac{1}{2}$. Hence $\int_{1/2}^1 |F_{X}^{-1}(1-u)| w(u)\d u\leq \int_{1/2}^1 F_{|X|}^{-1}(u) w(1-u)\d u = \int_0^{1/2} F_{|X|}^{-1}(1-u) w(u)\d u<\infty$. Since $u\mapsto |F_{X}^{-1}(1-u)|$ is increasing, we obtain that $X\in L_g$.

In the remaining case, there is some $\delta\in (0,\frac{1}{2}]$ such that $F_X^{-1}(1-u)\geq 0$ if $u\leq \delta$ and $F_X^{-1}(1-u)\leq 0$ if $u\geq 1-\delta$. It follows as above that $\int_{1-\delta}^1 |F_{X}^{-1}(1-u)| w(u)\d u<\infty$; also, $\int_{0}^\delta |F_{X}^{-1}(1-u)| w(u)\d u\leq \int_{0}^\delta F_{|X|}^{-1}(1-u) w(u)\d u<\infty$. This then implies again that $X\in L_g$.
\end{proof}

\begin{example}
(a) There is a decreasing weight $w$ such that $\Lambda(w)\subsetneq L_g$. Indeed, let $w(u)=2(1-u)$, $u\in [0,1]$. Also, choose a random variable $Y$ with $F_Y(x)=1-\frac{1}{x}$, $x\geq 1$, and take $X=-Y$. Then $|F_X^{-1}(1-u)|=\frac{1}{1-u}$, hence $X\in L_g$, but $F_{|X|}^{-1}(1-u)=\frac{1}{u}$, which shows that $X\notin \Lambda(w)$.

(b) There is a weight $w$ (which is necessarily not decreasing) such that $\Lambda(w)\not\subset L_g$. Indeed, let $w(u)=2u$, $u\in [0,1]$, and take the same random variable $X$ as in (a). Then $X\in \Lambda(w)$ but $X\notin L_g$.
\end{example}

In the same way, one can justify a claim made in Remark \ref{rem-OL}(a); for the notation we refer to Section \ref{s-distHaezGoov}.

\begin{proposition}\label{p-OrlLorentzLg}
Let $\phi$ be a Young function and $w$ decreasing. Then $\Lambda_{\phi,w}\subset L_g^\phi$.
\end{proposition}

On the other hand, the analogue of Proposition \ref{p-LorentzLgRhog} fails in general, even if $g$ is the identity.

\begin{example}\label{ex-LorentzLgRhog}
Let $g$ be the identity and $\phi(t)=t^2$. On $\Omega=(0,1]$ with the Lebesgue measure, we consider $X(\omega)=-\frac{1}{\sqrt{\omega}}$. Then $X\notin L_g^\phi$. But one calculates that, for $x\leq -2$, $\|(X-x)^+\|+x= \sqrt{x^2+4x+2\ln|x|+3}+x\geq -2$. 
\end{example}

We next prove a claim made in Remark \ref{r-Choquet}.

\begin{proposition}\label{p-rdu}
Let $\phi:\mathbb{R}\to\mathbb{R}$ be an increasing convex function with $\phi(0)=0$, $U(t)=-\phi(-t)$, $t\in\mathbb{R}$, the corresponding increasing concave function. Let $h$ be a distortion function with $h(0)=0$, and define $g(u)=1-h((1-u)-)$, $u\in [0,1]$. Then $g$ is a distortion function with $g(0)=0$ and, for any positive random variable $X$,
\[
(C)\int U(-X)\d (h\circ P) = - \int_0^1 \phi(F_{X}^{-1}(1-u))\d g(u),
\]
where the integral on the left is a Choquet integral.
\end{proposition}

\begin{proof} It is easy to see that $g$ is a distortion function with $g(u-)=1-h(1-u)$ on $[0,1]$.

For the notion of Choquet integrals, we refer to \cite{Den94} and \cite[p.\ 68]{Hei03}. Using a property of quantile functions and writing $Z=\phi(X)$, we see that it suffices to show that
\[
 (C)\int (-Z)\d (h\circ P) =  -\int_0^1 F_{Z}^{-1}(1-u)\d g(u).
\]
Since $Z\geq 0$, the Choquet integral equals
\[
\int_{-\infty}^0 (h(\overline{F}_{-Z}(x))-1) \d x = -\int_0^{\infty} (1-h(\overline{F}_{-Z}(-x))) \d x.
\]
Now, $\overline{F}_{-Z}(-x)= P(-Z>-x)=1-P(Z\geq x)$. For all but countably many $x$, this coincides with $1-P(Z>x)=1-\overline{F}_{Z}(x)$. For these $x$, we have
\[
1-h(\overline{F}_{-Z}(-x))= 1-h(1-\overline{F}_{Z}(x))=g(\overline{F}_{Z}(x)-).
\]
Hence, the Choquet integral equals
\[
-\int_0^{\infty} g(\overline{F}_{Z}(x)-) \d x = -\int_0^1 F_{Z}^{-1}(1-u)\d g(u),
\]
where we have used Proposition \ref{p-equiv}(a). This proves the claim.
\end{proof}

We finally give the proof of Theorem \ref{t-dHGtriv} (and hence of Theorem \ref{t-HGtriv}). For this we need two auxiliary results.

\begin{lemma}\label{l-triv1}
Let $g$ be a distortion function and $X\in L_g$. Then $X_n:=\max(\min(X,n),-n)\in L_g$, $n\geq 1$, and $\lim_{n\to\infty}\rho_g(X_n)= \rho_g(X)$.
\end{lemma}

\begin{proof} We have that $|F^{-1}_{X_n}|\leq |F^{-1}_{X}|$ and $F^{-1}_{X_n}\to F^{-1}_{X}$ on $(0,1]$. Thus the result follows from Definitions \ref{d-dist}, \ref{d-distrm}, and the dominated convergence theorem.
\end{proof}

\begin{lemma}\label{l-triv2}
Let $g$ be a distortion function with $g(0)=0$, $\phi$ a normalized Young function that satisfies the $\Delta_2$-condition, and $\alpha=0$. Let $X_n\in (L_g^\phi)^+$ with $X_n\searrow 0$. Then $\rho_{g,\phi,0}(X_n)\to 0$.
\end{lemma}

\begin{proof}
By Proposition \ref{p-distHaezGoov}(a) and the positivity of the $X_n$ it suffices to show that $a_n:=\pi_{g,\phi,0}(X_n)\to 0$.

Suppose, on the contrary, that $a:=\lim_n a_n>0$. Then, for any $n$, $a_n>0$, and by Proposition \ref{p-distOrlicz}(c) we get that $\int_0^1 \phi(\frac{F_{X_n}^{-1}(1-u)}{a_n}) \d g(u) = 1$, where we have used the $\Delta_2$-condition. Hence, by \eqref{eq-defdist} and a property of quantile functions,
\[
\int_0^{\infty} g\big(\overline{F}_{\phi(\frac{X_n}{a_n})}(x)-\big)\d x=1.
\]
On the other hand, since $\frac{X_n}{a_n}\to 0$, $\overline{F}_{\phi(\frac{X_n}{a_n})}(x)\to 0$ for all $x>0$. Since $g(0)=0$ and $g$ is continuous at $0$, the dominated convergence theorem implies that $0=1$; note that, by Lemma \ref{l-distOrlicz}(c), $\int_0^{\infty} g\big(\overline{F}_{\phi(\frac{X_1}{a})}(x)-\big)\d x<\infty$ by the $\Delta_2$-condition. This is the desired contradiction.
\end{proof}

\begin{proof}[Proof of Theorem \ref{t-dHGtriv}]
Let $X\in L_g^\phi$. For simplicity we write $\pi=\pi_{g,\phi,0}$ and $\sigma(x)=\pi((X-x)^+)+x$, $x\in \mathbb{R}$. By Proposition \ref{p-distHaezGoovinf}(a), $\rho_{g,\phi,0}(X)=\lim_{x\to-\infty} \sigma(x)$. 

The proof requires several steps.

(1) We first suppose that $X$ is bounded. 

(1a) Since $\phi$ is convex, it is left- and right-differentiable at $1$, so that $c_-$ and $c_+$ exist. Thus there is an increasing function $h:[0,\infty)\to[0,\infty)$ with $h(t)\to 0$ as $t\to 0$ such that, for $0\leq t \leq 1$,
\begin{equation}\label{eq-phi1}
0\leq\phi(t)-(1+c_-(t-1))\leq h(|t-1|)|t-1|,    
\end{equation}
and, for $t\geq 1$,
\begin{equation}\label{eq-phi2}
0\leq\phi(t)-(1+c_+(t-1))\leq h(|t-1|)|t-1|.
\end{equation} 

Next, let $x<\essinf X$. Then $P(X-x>0)=1$, hence $\pi(X-x)\neq 0$ by Proposition \ref{p-distOrlicz}(a); and since $X-x$ is bounded we have by Proposition \ref{p-distOrlicz}(c) that 
\[
\int_0^1 \phi\Big(\frac{F_{X-x}^{-1}(1-u)}{\pi(X-x)}\Big) \d g(u)= 1.
\]
Since $\sigma(x)-x=\pi((X-x)^+)= \pi(X-x)>0$, we have, using a property of quantile functions and the fact that $g(1-)=g(1)$, 
\begin{equation}\label{eq-phi3}
\int_{[0,1)} \phi\Big(\frac{F_X^{-1}(1-u)-x}{\sigma(x)-x}\Big) \d g(u)= 1.
\end{equation}
Also, since $X$ is bounded, $\sigma$ is increasing and $\sigma(t)$ converges as $t\to-\infty$, there is some $M>0$ such that $|F_X^{-1}|\leq M$ on $(0,1]$ and $|\sigma|\leq M$ on $(-\infty,0]$. 

Writing $t(u)= \frac{F_X^{-1}(1-u)-x}{\sigma(x)-x}$, we have that $t(u)-1= \frac{F_X^{-1}(1-u)-\sigma(x)}{\sigma(x)-x}$. Let $I_-=\{ u\in[0,1) : F_X^{-1}(1-u)\leq \sigma(x)\}$ and $I_+=[0,1)\setminus I_-$. Thus $t(u)\leq 1$ if and only if $u\in I_-$.

We now integrate $\phi(t(u))-(1+c_-(t(u)-1))$ over $I_-$ and $\phi(t(u))-(1+c_+(t(u)-1))$ over $I_+$, add the results, and apply \eqref{eq-phi1}, \eqref{eq-phi2}, \eqref{eq-phi3} and the fact that $\int_{[0,1)}\d g(u)=1$. We thus get, for $x<-M$,
\begin{align*}
-c_-\int_{I_-} \frac{F_X^{-1}(1-u)-\sigma(x)}{\sigma(x)-x}\,\d g(u)-&c_+\int_{I_+} \frac{F_X^{-1}(1-u)-\sigma(x)}{\sigma(x)-x}\,\d g(u)\\
& \leq h\Big(\frac{2M}{|x|-M}\Big) \int_{[0,1)} \Big|\frac{F_X^{-1}(1-u)-\sigma(x)}{\sigma(x)-x}\Big|\d g(u)\\
& \leq h\Big(\frac{2M}{|x|-M}\Big) \frac{2M}{\sigma(x)-x}.
\end{align*}
Writing $\delta:=c_+-c_-\geq 0$, and noting the definition of $\rho_g(X)$, we thus find that
\begin{align}\label{eq-deltaM}
c_-(\sigma(x)-\rho_g(X))+\delta \int_{I_+} (\sigma(x)-F_X^{-1}(1-u))\d g(u)\leq 2M.h\Big(\frac{2M}{|x|-M}\Big). 
\end{align}

We now distinguish two cases. 

(1b) Suppose that $X\geq 0$. Then $F_X^{-1}\geq 0$ on $(0,1]$; also, $\rho_{g,\phi,0}(X)\geq 0$ by monotonicity and hence $\sigma(x)\geq 0$ for all $x$ by Proposition \ref{p-distHaezGoovinf}. Thus, \eqref{eq-deltaM} implies that
\[
c_-(\sigma(x)-\rho_g(X))\leq\delta \rho_g(X) +2M.h\Big(\frac{2M}{|x|-M}\Big). 
\]
Letting $x\to-\infty$, we obtain that
\begin{equation}\label{eq-rhog+}
\rho_{g,\phi,0}(X)\leq \Big(\frac{\delta}{c_-}+1\Big)\rho_g(X)= \frac{c_+}{c_-}\rho_g(X).
\end{equation}

(1c) Now let $X\leq 0$, hence $F_X^{-1}\leq 0$ on $(0,1]$. Since $I_+=\varnothing$ if $\sigma(x)\geq 0$, we see that $\int_{I_+} (-\sigma(x))\d g(u) \leq  (-\sigma(x))^+$. Thus, \eqref{eq-deltaM} implies that
\[
c_-(\sigma(x)-\rho_g(X))\leq\delta (-\sigma(x))^+ +2M.h\Big(\frac{2M}{|x|-M}\Big). 
\]
Letting $x\to-\infty$, and noting that $\rho_{g,\phi,0}(X)\leq 0$, we obtain that
\begin{equation}\label{eq-rhog-}
\rho_{g,\phi,0}(X)\leq\frac{c_-}{c_-+\delta}\rho_g(X)= \frac{c_-}{c_+}\rho_g(X).
\end{equation}

(1d) Finally, for arbitrary bounded $X$, we write $X=X^+-X^-$. Since $X^+=\max(X,0)$ and $-X^-=\min(X,0)$ are comonotonic, Proposition \ref{p-distHGcom}, \eqref{eq-rhog+} and \eqref{eq-rhog-}, with Proposition \ref{p-distHaezGoov}(b), imply that
\begin{align*}
\rho_g(X)\leq\rho_{g,\phi,0}(X)\leq \rho_{g,\phi,0}(X^+) + \rho_{g,\phi,0}(-X^-)\leq \frac{c_+}{c_-}\rho_g(X^+) + \frac{c_-}{c_+}\rho_g(-X^-).
\end{align*}
This shows the desired inequality for $X\in L^\infty$.

(2) We now let $X\in L_g^\phi$ be arbitrary, where we assume that $\phi$ satisfies the $\Delta_2$-condition.

(2a) Suppose again that $X\geq 0$. Assume first that $g(0)=0$. Since $X=\min(X,n)+(X-n)^+$, and since $\min(X,n)$ and $(X-n)^+$ are comonotonic, it follows from Proposition \ref{p-distHGcom} that $\rho_{g,\phi,0}(X)\leq \rho_{g,\phi,0}(\min(X,n))+\rho_{g,\phi,0}((X-n)^+)$, hence, by \eqref{eq-rhog+},
\[
\rho_{g,\phi,0}(X)\leq \frac{c_+}{c_-}\rho_{g}(\min(X,n))+\rho_{g,\phi,0}((X-n)^+).
\]
Letting $n\to\infty$, and applying Lemmas \ref{l-triv1} and \ref{l-triv2}, we obtain that
\begin{equation}\label{eq-rhog+2}
\rho_{g,\phi,0}(X)\leq \frac{c_+}{c_-}\rho_g(X).
\end{equation}
On the other hand, suppose that $g(0)>0$. Then $X$ is bounded above, hence bounded, by the discussion after Definition \ref{d-OL}, so that \eqref{eq-rhog+2} holds by \eqref{eq-rhog+}.

(2b) Now suppose that $X\leq 0$. Then, for $n\geq 1$,
\[
\rho_{g,\phi,0}(X)\leq \rho_{g,\phi,0}(\max(X,-n)).
\]
Applying \eqref{eq-rhog-}, we get that
\[
\rho_{g,\phi,0}(X)\leq \frac{c_-}{c_+}\rho_g(\max(X,-n)).
\]
Letting $n\to\infty$, and applying Lemma \ref{l-triv1}, we obtain that
\begin{equation}\label{eq-rhog-2}
\rho_{g,\phi,0}(X)\leq \frac{c_-}{c_+}\rho_g(X).
\end{equation}

(2c) One can now obtain the desired inequality for arbitrary $X\in L_g^\phi$ as in (1d), using this time \eqref{eq-rhog+2} and \eqref{eq-rhog-2}.
\end{proof}

\begin{proof}[Proof of Corollary \ref{c-dHGtriv}]
If $\phi$ is differentiable at $1$ then $c_-=c_+$ in Theorem \ref{t-dHGtriv} and hence
\[
\rho_g(X)\leq \rho_{g,\phi,0}(X)\leq \rho_g(X^+) + \rho_g(-X^-).
\]
Since $X^+$ and $-X^-$ are comonotonic, Proposition \ref{p-distcomon} implies that the right-hand side equals
$\rho_g(X^+-X^-)=\rho_g(X)$, so that the result follows.
\end{proof}

\subsection*{Acknowledgements} The authors are grateful to Anna Kami\'nska for sharing her expertise on Orlicz-Lorentz spaces. We also thank Emanuela Rosazza Gianin and the members of the PhD-committee of the first author, Michèle Vermaele, Jan Dhaene, and Dani\"el Linders, for interesting discussions. We are grateful to Dani\"el Linders for suggesting the link with risk-dependent expected utilities, and to an anonymous editor for pointing out the link with return risk measures and for suggesting robustification.


\begin{thebibliography}{99}
\bibitem{AcTa02} C. Acerbi and D. Tasche, On the coherence of expected shortfall, \emph{J. Banking \& Finance} 26 (2002), 1487--1503.

\bibitem{AhSh14} J.~Y. Ahn and N.~D. Shyamalkumar, Asymptotic theory for the empirical Haezendonck-Goovaerts risk measure, \emph{Insurance Math. Econom.} 55 (2014), 78--90.

\bibitem{AmLi24} M. Amarante and F.-B. Liebrich, Distortion risk measures: prudence, coherence, and the expected shortfall, \emph{Math. Finance} 34 (2024), 1291--1327.

\bibitem{ADEH99} P. Artzner, F. Delbaen, J.-M. Eber, and D. Heath, Coherent measures of risk, \emph{Math. Finance} 9 (1999), 203--228.

\bibitem{ABL23} M. Ayg\"un, F. Bellini, and R. J. A. Laeven, Elicitability of return risk measures,	arXiv:2302.13070 (2023).

\bibitem{ABL24} M. Ayg\"un, F. Bellini, and R. J. A. Laeven, On geometrically convex risk measures, arXiv:2403.06188 (2024).

\bibitem{ABL25} M. Ayg\"un, F. Bellini, and R. J. A. Laeven, Generalized Orlicz premia, arXiv:2507.09181 (2025).

\bibitem{BLR18} F. Bellini, R. J. A. Laeven, and E. Rosazza Gianin, Robust return risk measures, \emph{Math. Financ. Econ.} 12 (2018), 5--32.

\bibitem{BeRo08} F. Bellini and E. Rosazza Gianin, On Haezendonck risk measures, \emph{J. Banking \& Finance} 32 (2008), 986--994.

\bibitem{BeRo08b} F. Bellini and E. Rosazza~Gianin, Optimal portfolios with Haezendonck risk measures, \emph{Statist. Decisions} 26 (2008), 89--108.

\bibitem{BeRo12} F. Bellini and E. Rosazza Gianin, Haezendonck-Goovaerts risk measures and Orlicz quantiles, \emph{Insurance Math. Econom.} 51 (2012), 107--114.

\bibitem{CCR21} G. Canna, F. Centrone, and E. Rosazza~Gianin, Haezendonck-Goovaerts capital allocation rules, \emph{Insurance Math. Econom}. 101 (2021), 173--185.

\bibitem{CRS07} M. J. Carro, J. A. Raposo, and J. Soria, Recent developments in the theory of Lorentz spaces and weighted inequalities, 
\emph{Mem. Amer. Math. Soc.} 187 (2007), no. 877.

\bibitem{CCM97} A. Chateauneuf, M. Cohen, and I. Meilijson, Comonotonicity, rank-dependent utilities and a search problem, in: \textit{ Distributions with given marginals and moment problems} (Prague, 1996), pp.\ 73--79, Kluwer, Dordrecht 1997.

\bibitem{Che96} S. Chen, Geometry of Orlicz spaces, \textit{Dissertationes Math. (Rozprawy Mat.)} 356 (1996).

\bibitem{CGLL22} S. Chen, N.  Gao, D. H. Leung, and L. Li, Automatic Fatou property of law-invariant risk measures, \emph{Insurance Math. Econom.} 105 (2022), 41--53.

\bibitem{ChRe25} J. Chudziak and P. Rela, The Orlicz premium principle under uncertainty, \emph{Rev. R. Acad. Cienc. Exactas Fís. Nat. Ser. A Mat. RACSAM} 119 (2025), Paper No. 116, 23 pp.

\bibitem{Den94} D. Denneberg, \emph{Non-additive measure and integral}, Kluwer, Dordrecht 1994.

\bibitem{DDGK05} M. Denuit, J. Dhaene, M. Goovaerts, and R. Kaas, \textit{Actuarial theory for dependent risks}, John Wiley \& Sons, 
Chichester 2005.

\bibitem{DDGKL06} M. Denuit, J. Dhaene, M. Goovaerts, R. Kaas, and R. Laeven, Risk measurement with equivalent utility principles, \textit{Statist. Decisions} 24 (2006), 1--25.

\bibitem{DDGKV02} J. Dhaene, M. Denuit, M. J. Goovaerts, R Kaas, and D. Vyncke, The concept of comonotonicity in actuarial science and finance: theory, \textit{Insurance Math. Econom.} 31 (2002), 3--33.

\bibitem{DKLT12} J. Dhaene, A. Kukush, D. Linders, and Q. Tang, Remarks on quantiles and distortion risk measures, \emph{Eur. Actuar. J.} 2 (2012), 319--328.

\bibitem{DLVDG08} J. Dhaene, R. J. A. Laeven, S. Vanduffel, G. Darkiewicz, and M. J. Goovaerts, Can a coherent risk measure be too subadditive, \emph{J. Risk Insurance} 75 (2008), 365--386.

\bibitem{DVGKTV06} J. Dhaene, S. Vanduffel, M. J. Goovaerts, R. Kaas, Q. Tang, and D. Vyncke, Risk measures and comonotonicity: a review, \emph{Stoch. Models} 22 (2006), 573--606.

\bibitem{EdSo92} G. A. Edgar and L. Sucheston, \textit{Stopping times and directed processes}, Cambridge University Press, Cambridge 1992.

\bibitem{EmWa15} P. Embrechts and R. Wang, Seven proofs for the subadditivity of expected shortfall, \emph{Depend. Model.} 3 (2015), 126--140.

\bibitem{FoSc16} H. F\"ollmer and A. Schied, \textit{Stochastic finance}, fourth edition, De Gruyter, Berlin 2016. 

\bibitem{FrWi24} C. Fröhlich and R. C. Williamson, Risk measures and upper probabilities: coherence and stratification, \emph{J. Mach. Learn. Res.} 25 (2024), Paper No. 207, 100 pp.

\bibitem{GMX20} N. Gao, C. Munari, and F. Xanthos, Stability properties of Haezendonck–Goovaerts premium principles, \emph{Insurance Math. Econom.} 94 (2020), 94--99.

\bibitem{GeGe25} H. Geiss and S. Geiss, \textit{Measure, probability and functional analysis}, Springer, Cham 2025.

\bibitem{GoPe20} C. Gonzales and P. Perny, Decision under uncertainty, in: P. Marquis, O. Papini, and H. Prade (editors), \emph{A guided tour of artificial intelligence research.\ Vol.\ I.\ Knowledge representation, reasoning and learning}, pp.\ 549--586,
Springer, Cham 2020.

\bibitem{GKDT04} M. J. Goovaerts, R. Kaas, J. Dhaene, and Q. Tang, Some new classes of consistent risk measures, \emph{Insurance Math. Econom.} 34 (2004), 505--516.

\bibitem{GLVT12} M. Goovaerts, D. Linders, K. Van Weert, and F. Tank, On the interplay between distortion, mean value and Haezendonck-Goovaerts risk measures, \emph{Insurance Math. Econom.} 51 (2012), 10--18.

\bibitem{Gou22} A. Goulard, \textit{Les mesures de risque de Haezendonck-Wang}, PhD-thesis, Université de Mons, Mons 2022. 

\bibitem{HaGo82} J. Haezendonck and M. Goovaerts, A new premium calculation principle based on Orlicz norms, 
\emph{Insurance Math. Econom.} 1 (1982), 41--53.

\bibitem{HJZ15} X.~D. He, H. Jin, and X.~Y. Zhou, Dynamic portfolio choice when risk is measured by weighted VaR, \emph{Math. Oper. Res.} 40 (2015), 773--796.

\bibitem{Hei03} S. Heilpern, A rank-dependent generalization of zero utility principle, \emph{Insurance Math. Econom.} 33 (2003), 67--73.

\bibitem{HKM02} H. Hudzik, H., A. Kami\'nska, and M. Masty{\l}o, On the dual of Orlicz-Lorentz space, \emph{Proc. Amer. Math. Soc.} 130 (2002), 1645--1654.

\bibitem{JST06} E. Jouini, W. Schachermayer, and N. Touzi, Law invariant risk measures have the Fatou property, in: \emph{Advances in mathematical economics. Vol. 9}, pp. 49–71, Springer, Tokyo 2006.

\bibitem{KDG000} R. Kaas, J. Dhaene, and M.~J. Goovaerts, Upper and lower bounds for sums of random variables, \emph{Insurance Math. Econom.} 27 (2000), 151--168.

\bibitem{Kam90} A. Kami\'nska, Some remarks on Orlicz-Lorentz spaces, \emph{Math. Nachr.} 147 (1990), 29--38.

\bibitem{KMP03} A. Kami\'nska, L. Maligranda, and L. E. Persson, Indices, convexity and concavity of Calderón-Lozanovskii spaces,
\emph{Math. Scand.} 92 (2003), 141--160.

\bibitem{LRZ24} R.~J.~A. Laeven, E. Rosazza~Gianin, and M. Zullino, Law-invariant return and star-shaped risk measures, \textit{Insurance Math. Econom.} 117 (2024), 140--153.

\bibitem{LRZ26} R.~J.~A. Laeven, E. Rosazza~Gianin, and M. Zullino, Dynamic return and star-shaped risk measures via BSDEs, preprint (2026).

\bibitem{LLS25} C. Laudagé, F.-B. Liebrich, and J. Sass, Multi-asset return risk measures, \textit{ASTIN Bull.} 55 (2025), 668--694. 

\bibitem{LiMu25} F.-B. Liebrich and C. Munari, Short communication: revisiting the automatic Fatou property of law-invariant functionals, \emph{SIAM J. Financial Math.} 16 (2025), SC1--SC11.

\bibitem{LPW17} Q. Liu, L. Peng, and X. Wang, Haezendonck-Goovaerts risk measure with a heavy tailed loss, \emph{Insurance Math. Econom.} 76 (2017), 28--47.

\bibitem{Lor51} G. G. Lorentz, On the theory of spaces $\Lambda$, \emph{Pacific J. Math.} 1 (1951), 411--429.

\bibitem{Meg98} R.~E. Megginson, \textit{An introduction to Banach space theory}, Springer, New York 1998.

\bibitem{Pic13} A. Pichler, The natural Banach space for version independent risk measures, \emph{Insurance Math. Econom.} 53 (2013), 405--415.

\bibitem{PKJF13} L. Pick, A. Kufner, O. John, and S. Fu\v{c}\'{\i}k, \emph{Function spaces, Vol. 1,} second edition, Walter de Gruyter \& Co., Berlin 2013. 

\bibitem{Qui82} J. Quiggin, A theory of anticipated utility, \emph{J. Econ. Behav. Organization} 3 (1982), 323--343.

\bibitem{Qui93} J. Quiggin, \emph{Generalized expected utility theory: The rank-dependent model}, Kluwer, Boston 1993.

\bibitem{RaRe91} M. M. Rao and Z. D. Ren, \emph{Theory of Orlicz spaces}, Marcel Dekker, New York 1991.

\bibitem{Roc74} R. T. Rockafellar, \emph{Conjugate duality and optimization}, Society for Industrial and Applied Mathematics, Philadelphia, PA 1974.

\bibitem{RoUr00} R. T. Rockafellar and S. Uryasev, Optimization of conditional value-at-risk, \emph{J. Risk} 2 (2000), no. 2, 21--41.

\bibitem{RoUr02} R. T. Rockafellar and S. Uryasev, Conditional value-at-risk for general loss distributions, \emph{J. Banking \& Finance} 26 (2002), 1443--1471.

\bibitem{RoSg13} E. Rosazza Gianin and C. Sgarra, \emph{Mathematical finance: theory review and exercises}, Springer, Cham 2013.

\bibitem{Rue13} L. R\"uschendorf, \textit{Mathematical risk analysis}, Springer, Heidelberg 2013.

\bibitem{ShSh07} M. Shaked and J. G. Shanthikumar, \textit{Stochastic orders}, Springer, New York 2007.

\bibitem{Svi10} G. Svindland, Continuity properties of law-invariant (quasi-)convex risk functions on $L^\infty$, \emph{Math. Financ. Econ.} 3 (2010), 39--43.

\bibitem{TaYa14} Q. Tang and F. Yang, Extreme value analysis of the Haezendonck-Goovaerts risk measure with a general Young function, \emph{Insurance Math. Econom.} 59 (2014), 311--320.

\bibitem{Wak10} P.~P. Wakker, \emph{Prospect theory: For risk and ambiguity}, Cambridge University Press, Cambridge 2010.

\bibitem{Wa96} S. Wang, Premium calculation by transforming the layer premium density, \emph{ASTIN Bull.} 26 (1996), 71--92.

\bibitem{WaDh98} S. Wang and J. Dhaene, Comonotonicity, correlation order and premium principles, \emph{Insurance Math. Econom.} 22 (1998), 235--242.

\bibitem{WaXu23} W. Wang and H. Xu, Preference robust distortion risk measure and its application, \emph{Math. Finance} 33 (2023), 389--434.

\bibitem{Wei18} P. Wei, Risk management with weighted VaR, \emph{Math. Finance} 28 (2018), 1020--1060.

\bibitem{WuXu22} Q. Wu and H. Xu, Robust distorted Orlicz premium: modelling, computational scheme and applications, \textit{J. Risk} 27 (2025), no. 6, 1--49. Also SSRN: http://dx.doi.org/10.2139/ssrn.4093580 (April 26, 2022).

\bibitem{Zaa83} A. C. Zaanen, \textit{Riesz spaces. II}, North-Holland Publishing Co., Amsterdam 1983.
\end{thebibliography}
\end{document}